\documentclass[12pt]{article}
\usepackage{graphicx,amssymb,amsmath,empheq}
\usepackage{authblk}
\usepackage{amsthm}
\usepackage{amsfonts}
\usepackage{amssymb}
\usepackage{dsfont}
\usepackage{bm}
\usepackage{empheq}
\usepackage{subfigure}
\usepackage{braket}
\usepackage{siunitx}
\usepackage{amsmath,amssymb,amsfonts,stmaryrd}
\usepackage{hyperref}
\hypersetup{colorlinks = true, linkcolor=purple, citecolor=blue, urlcolor=blue}
\usepackage{url}
\usepackage{dsfont}
\usepackage{comment}
\usepackage{svg}
\usepackage{graphicx}
\usepackage{physics}
\usepackage{makecell}
\usepackage{rotating}
\usepackage{multirow} 
\usepackage[normalem]{ulem}

\theoremstyle{plain}

\usepackage{tikz}
\usepackage{tikz-cd}
\usetikzlibrary{arrows}
\usetikzlibrary{intersections}
\usetikzlibrary{shapes.geometric}
\usetikzlibrary{decorations.pathmorphing, patterns,shapes}
\usetikzlibrary{decorations.markings}

\tikzset{
  mid arrow/.style={postaction={decorate,decoration={
        markings,
        mark=at position .575 with {\arrow[#1]{stealth}}
      }}},
  near arrow/.style={postaction={decorate,decoration={
        markings,
        mark=at position .275 with {\arrow[#1]{stealth}}
      }}},
   far arrow/.style={postaction={decorate,decoration={
        markings,
        mark=at position .800 with {\arrow[#1]{stealth}}
      }}},
}

\usepackage[nosort]{cite}

\newtheorem{theorem}{Theorem}
\newtheorem*{theorem*}{Theorem}
\newtheorem{proposition}{Proposition}
\newtheorem{lemma}{Lemma}
\newtheorem{definition}{Definition}

\newtheorem*{corollary*}{Corollary}

\newcommand{\dkap}{\delta\kern-1.25pt\varkappa}

\newcommand{\R}{{\mathrm{R}}}
\newcommand{\A}{{\mathrm{A}}}

\newcommand{\onenorm}[1]{\lVert #1 \rVert_{1}}

\newcommand{\calC}{\mathcal{C}}
\newcommand{\calD}{\mathcal{D}}

\newcommand{\calF}{\mathcal{F}}
\newcommand{\calG}{\mathcal{G}}
\newcommand{\calH}{\mathcal{H}}
\newcommand{\calI}{\mathcal{I}}

\newcommand{\calL}{\mathcal{L}}
\newcommand{\calM}{\mathcal{M}}

\newcommand{\calP}{\mathcal{P}}

\newcommand{\calS}{\mathcal{S}}

\newcommand{\calU}{\mathcal{U}}

\newcommand{\st}{\operatorname{st}}

\renewcommand{\H}{\operatorname{H}}
\renewcommand{\epsilon}{\varepsilon}

\newcommand{\poly}{\operatorname{poly}}

\newcommand{\Lr}{\operatorname{LR}}
\newcommand{\Lgr}{\operatorname{LGR}}
\newcommand{\GR}{\operatorname{GR}}
\newcommand{\LR}{\mathsf{LR}}
\newcommand{\SM}{\operatorname{SM}}

\newcommand{\SP}{\mathcal{SP}}

\newcommand{\LU}{\operatorname{LU}}
\newcommand{\LF}{\operatorname{LF}}
\newcommand{\nA}{\operatorname{nA}}
\newcommand{\supp}{\operatorname{supp}}
\newcommand{\Om}{\Omega}
\newcommand{\IR}{\operatorname{IR}}
\newcommand{\MERA}{\operatorname{MERA}}
\renewcommand{\mod}{\text{ mod }}
\newcommand{\UFC}{\operatorname{UFC}}
\newcommand{\dist}{\operatorname{dist}}

\newcommand{\boldu}{\boldsymbol{u}}
\newcommand{\boldv}{\boldsymbol{v}}

\newcommand{\bolda}{\boldsymbol{a}}
\newcommand{\boldb}{\boldsymbol{b}}
\newcommand{\boldc}{\boldsymbol{c}}
\newcommand{\boldd}{\boldsymbol{d}}

\newcommand{\Ens}{\mathbb{E}}

\newcommand{\YZ}[1]{\textcolor{blue}{YZ: #1}}

\makeatletter

\newcommand*{\wideboxed}[1]{\setlength{\fboxsep}{1ex}%
  \fbox{\m@th$\displaystyle#1$}}
\makeatother

\title{
Extensive long-range magic\\ in non-Abelian topological orders
}

\author[1]{Yuzhen Zhang\thanks{yuzhenzhangph@gmail.com}}
\author[2]{Isaac H. Kim\thanks{ikekim@ucdavis.edu}}
\author[3]{Yimu Bao\thanks{baoyimu@gmail.com}}
\author[1]{Sagar Vijay\thanks{sagarvijay@ucsb.edu}}

\affil[1]{\normalsize\it Department of Physics, University of California, Santa Barbara, CA 93106, USA}
\affil[2]{\normalsize\it Department of Computer Science, University of California, Davis, CA 95616, USA}
\affil[3]{\normalsize\it Kavli Institute for Theoretical Physics, Santa Barbara, CA 93106, USA}

\date{\today}

\begin{document}

\maketitle
\begin{abstract}
We show that the low-energy states of non-Abelian topological orders possess extensive magic which is long-ranged, and cannot be eliminated by a constant-depth local unitary circuit.  This refines conventional notions of complexity beyond the linear circuit depth which is required to prepare any topological phase, and provides a new resource-theoretic characterization of topological orders. A central technical result is a no-go theorem establishing that stabilizer states—even up to constant-depth local unitaries—cannot approximate low-energy states of non-Abelian string-net models which satisfy the entanglement bootstrap axioms. Moreover, we show that stabilizer-realizable Abelian string-net phases have mutual braiding phases quantized by the on-site qudit dimension, and that any violation of this condition necessarily implies extensive long-range magic.
Extending to higher spatial dimensions, we argue that any state obeying an entanglement area law and hosting excitations with nontrivial fusion spaces must exhibit extensive long-range magic. This applies, in particular, to ground-states and low-energy states of higher-dimensional quantum double models.
\end{abstract}

\newpage
\tableofcontents
\newpage

\section{Introduction}

Topological quantum order characterizes zero-temperature phases of matter with robust long-range entanglement and emergent anyon excitations~\cite{Wen_2017}. All topologically-ordered states are complex in a coarse sense: they cannot be prepared from a product state by a finite-depth local unitary circuit, but instead require circuit depth at least linear in system size~\cite{Bravyi:2006zz}. This criterion, however, does not capture the fine-grained complexity of topological quantum matter. A broad class of Abelian topological orders can be realized as stabilizer codes~\cite{Ellison:2021vth}, and are therefore classically simulable despite the linear circuit depth needed to prepare these states. Recent work has shown that certain other topological orders, including some that host non-Abelian anyons, can be efficiently prepared using local measurements and classical feedforward~\cite{Briegel:2000unx,Raussendorf:2005dmx,Aguado:2008zz,Piroli:2021fjn,Tantivasadakarn:2021vel,verresen2021efficiently,Bravyi:2022zcw,Tantivasadakarn:2022hgp,Iqbal:2023shx,Foss-Feig:2023uew}. These observations point to the need for finer notions of quantum many-body complexity. A natural candidate is ``magic": the non-Clifford resource required to prepare, represent, or efficiently simulate a many-body quantum state.

In this work, we investigate whether non-Abelian topological order necessarily carries intrinsic magic. In non-Abelian topological order, information can be stored nonlocally in fusion space, and braiding can implement logical operations, sometimes including non-Clifford gates~\cite{Nayak_2008}. We ask whether the presence of such fusion structure underlying the computational power of non-Abelian anyons is already visible at the level of the many-body wavefunction as an intrinsic non-Clifford resource. Since magic is an inherently basis-dependent notion, defined only after specifying a local choice of Pauli operators and Clifford unitaries, the central question is whether non-Abelian topological order necessarily supports magic that cannot be removed by any local change of basis. This question is particularly relevant for non-Abelian topological codes~\cite{Koenig:2010uua,Wootton:2014,Dauphinais:2016loh,Cui:2019lvb}. In this setting, it is natural to ask whether preparing or simulating these codes necessarily requires significant non-Clifford resources. More broadly, investigating basis-independent magic is motivated by the search for sharper notions of low-energy complexity in quantum matter.% and by possible connections to quantum complexity theory.

This broader motivation naturally connects to questions in quantum complexity theory. It is natural to ask whether magic that is robust under local basis changes constrains the complexity of low-energy many-body states. This perspective is closely related to the quantum PCP  (probabilistically checkable proofs) conjecture in quantum complexity theory, which asks whether estimating the ground-state energy of a local Hamiltonian to constant accuracy remains QMA-hard~\cite{Aharonov:2013hwe}. A necessary prerequisite for such hardness is the no-low-energy trivial-state (NLTS) property~\cite{Freedman:2013zfj}: local Hamiltonians where all sufficiently low-energy states are long-range entangled. While NLTS is necessary, it is not sufficient for quantum PCP because long-range entanglement alone does not capture all forms of computational hardness. For instance, recent NLTS constructions~\cite{Anshu:2022hsn} are stabilizer Hamiltonians; such systems have an abundance of low-energy states that are stabilizer states. Stronger notions have also been studied, including Hamiltonians with no low-energy stabilizer states~\cite{Coble:2023uks} and Hamiltonians whose low-energy states require $\Omega(n)$ $T$-gates~\cite{Coble:2023hys}. Yet these notions still allow the possibility of low-energy states that are stabilizer states up to a constant-depth local unitary circuit, which makes it easy to estimate the energy of these states to constant accuracy~\cite{Wei:2025irp}. This motivates a stronger question: can one find Hamiltonians whose low-energy states have parametrically large magic that persists under arbitrary constant-depth local unitaries? This form of ``no low-energy trivial magic" could provide further evidence for quantum PCP.% \YZ{``NLTM''' was dubbed in the paper~\cite{Wei:2025irp}, where they donnot require \emph{large} long-range magic at low-energies, but only the \emph{existence} of long-range magic. But I think the mere existence of LRM would be a weak condition. What shall we say here?}

Before answering these questions, we formalize the idea of ``magic modulo local basis rotations''. If a state can be prepared by applying a constant-depth local unitary to a stabilizer state, then the state has no magic in a locally rotated basis. Such states are called \emph{short-range magic} (SRM) states, depicted in Fig.~\ref{fig:SRM}.
\begin{figure}
    \centering
    \includegraphics[width=0.6\linewidth]{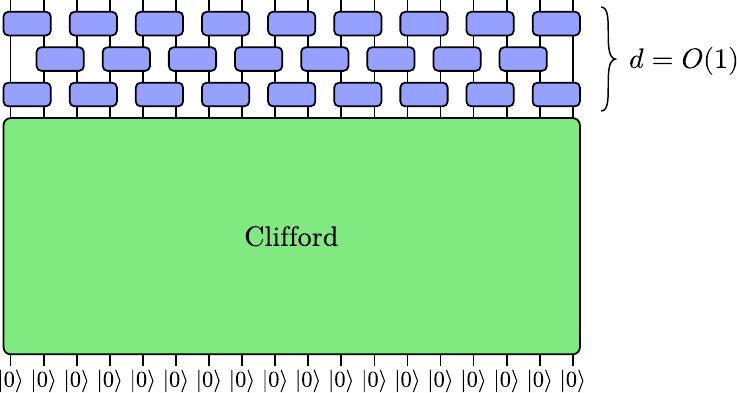}
    \caption{Circuit architecture depicting a one-dimensional qudit array with short-range magic.}
    \label{fig:SRM}
\end{figure}
Such states nevertheless have a list of interesting properties:
\begin{itemize}
\item SRM states can be long-range entangled. Applying a constant-depth local unitary to an arbitrary long-range entangled stabilizer state gives a SRM state with long-range entanglement.
\item SRM states can have extensive magic. For example, a tensor product of $T$ states have extensive magic for the magic measures considered below. 

\item Local observables can be efficiently simulated on SRM states. After conjugating a local observable by the constant-depth local unitary, only gates in its lightcone contribute, so the evolved observable remains supported on a constant-size region. One can then decompose it into Pauli operators and compute each Pauli expectation value on the underlying stabilizer state. As a consequence, the energy of an SRM state with respect to any local Hamiltonian can be efficiently evaluated. This is crucial for the quantum PCP conjecture~\cite{Aharonov:2013hwe}, because if a Hamiltonian has low-energy SRM states, then it cannot be a candidate for quantum PCP~\cite{Wei:2025irp}.
\item SRM states can also be highly random. In fact, ensembles of SRM states can reproduce the first $k$ moments of Haar random states up to relative error for $k=O(1)$~\cite{Zhang:2025dhg}.
\end{itemize}

To isolate the basis-independent part of magic that captures genuinely long-distance structure, one should minimize magic by performing local basis rotations and look at the remaining magic structure. To formalize this idea, we define \emph{long-range magic} (LRM) as the minimal remaining magic after applying an arbitrary finite depth local unitary:
\begin{equation}
\LR M_d(\rho)=\min_{U\in\LU_d}M(U\rho U^\dagger),
\end{equation}
where $M(\rho)$ is some magic measure, and $\LU_d$ is the set of local unitaries with depth at most $d$. These measures of magic qualitatively relate to the distance to a stabilizer state, and will be formally defined in Sec. \ref{subsec:LRM_magic_measures}. As a scale, an array of $T$ states has $M(\ketbra{T}^{\otimes n})=O(n)$ for the relevant extensive magic measures, while these measures are at most $O(n)$ on an $n$-qudit system. Closely related notions of long-range magic have recently appeared in several contexts, where the \emph{existence} of long-range magic has been shown~\cite{Korbany:2025noe,Wei:2025irp,Parham:2025sxj,Li:2026tsy}. The goal of this work is to give a \emph{quantitative} evaluation using the LRM measure defined above.

\subsection{Overview of Results}
Our main results are the following theorems. Our first result is that non-Abelian topological order carries \emph{unconditional} extensive long-range magic: the conclusion does not depend on choosing a special local Hilbert-space dimension.
\begin{theorem*}
The ground state $\ket{\psi_{\nA}}$ of a non-Abelian string-net model has long-range magic
\begin{equation}
\LR M_d(\ketbra{\psi_{\nA}})\ge n\cdot \calC_{\nA}(\UFC,q,d),
\end{equation}
in a system of $n$ qudits with qudit dimension $q$, where $\calC_{\nA}(\UFC,q,d)$ only depends on the unitary fusion category (UFC), $q$ and circuit depth $d$.
\end{theorem*}

Furthermore, the extensiveness of long-range magic persists at low energy densities, for sufficiently short circuit depth:
\begin{theorem*}
For a non-Abelian string-net Hamiltonian $H$ on $n$ qudits, where each term in the Hamiltonian has $O(1)$ operator norm, any state $\ket{\psi}$ with energy $\bra{\psi}H\ket{\psi}$ less than $\epsilon n$ has long-range magic
\begin{equation}
\LR\Lr_d(\ketbra{\psi})\ge\Omega(nd^{-2}-n\epsilon^{1-\alpha}\log(1/\epsilon))\cdot\calC_{\nA}(\UFC,q,d).
\end{equation}
\end{theorem*}
Here $\Lr$ is a specific magic measure---log robustness of magic---which will be formally defined in Sec.~\ref{subsec:LRM_magic_measures}.

These results extend to higher dimensions. In $D$ spatial dimensions, any state $\ket{\psi}$ that satisfies strict entanglement area law and hosts emergent excitations with a non-trivial fusion space has long-range magic $\LR M_d(\ketbra{\psi})\ge n\cdot \calC(q,d)$.\footnote{A state is said to satisfy a strict area law if the entanglement entropy of a ball has a leading contribution proportional to the boundary area and a universal subleading term. This implies the entanglement bootstrap axioms \textbf{A1} in $D=2$~\cite{kim2024strictarealawimplies} and $D=3$~\cite{Huang:2021gsc,Shi_2025}.} An example would be the ground state of the $S_3$ quantum double in three spatial dimensions. Moreover, in this example, the low-energy states have long-range magic $\Lr_d(\ketbra{\psi})=\Omega(nd^{-3}-n\epsilon^{1-\alpha}\log(1/\epsilon))\cdot\calC(q,d).$ Here, the excitations and the dimension of the fusion space are defined by the information convex set in an appropriate topology, as previously studied in Ref.~\cite{Shi:2019mlt}, and as we review in later sections. 

The central step in demonstrating the long-range magic of non-Abelian string-net models is to show that stabilizer states can realize only Abelian string-net phases. This statement was proven in~\cite{Bombin:2013ris,Haah:2020lnd} under the assumptions of translation invariance and prime qudit dimensions, but beyond these settings it has remained a conjecture. Here, we establish this result rigorously without imposing any additional assumptions:  
\begin{theorem*}
Stabilizer states, even up to constant-depth local unitaries, cannot realize the ground states of non-Abelian string-net models.
\end{theorem*}

Interestingly, stabilizer states can nonetheless host non-Abelian \emph{defects}. For instance, it is possible to introduce a line-like lattice dislocation in the toric code, whose endpoints exhibit non-Abelian braiding statistics~\cite{Bombin2010}. Therefore, while the existence of non-Abelian anyons implies there is long-range magic, the existence of non-Abelian defects does not. We note that defects and anyons are different objects~\cite{Barkeshli2019}, and as such, there is no contradiction. In particular, the state hosting the non-Abelian defects cannot be connected to any string-net wavefunction by any constant-depth circuit.\footnote{More precisely, if the circuit depth is sufficiently small compared to the minimal distance between the endpoints, the circuit cannot convert one state to the other.}

%\IK{The subtle difference between the two states can be understood as follows. From a quantum information perspective, states hosting such defects cannot be connected to any string-net wavefunction by a constant-depth circuit; if the circuit depth is small compared to the minimal distance between the endpoints, there is an obstruction to converting one state to the other. From the perspective of topological order, the anyons correspond the objects in unitary modular tensor category. In contrast, the non-Abelian defects can be understood as objects in $G$-crossed unitary braided fusion category~\cite{Barkeshli2019}.}

%Moreover, we
We also show that stabilizer states can only realize the ground state of a restricted set of Abelian string-net models, which are constrained by the qudit dimension $q$. Moreover, if the qudit dimension $q$ violates these constraints, then the string-net model must have extensive long-range magic. 
\begin{theorem*}
Stabilizer states can only realize Abelian string-net models where mutual braiding phases are $q$-th roots of unity. Moreover, if there exist mutual braiding phases that are not $q$-th roots of unity, then the Abelian string-net ground state has long-range magic
\begin{equation}
\LR M_d(\ketbra{\psi_{\A}})\ge n\cdot \calC_{\A}(\UFC,q,d),
\end{equation}
where $\calC_{\A}(\UFC,q,d)$ only depends on the unitary fusion category, $q$, and circuit depth $d$. As a concrete example, this means that qubit stabilizer states cannot realize a $\mathbb{Z}_{3}$ topological quantum order, though this is certainly possible with qudit stabilizer codes with dimension $q=3$.  %Another example is that qubit stabilizer states cannot realize a double-semion topological quantum order, though this is possible with qudit stabilizer codes with dimension $q=4$.
\end{theorem*}
We note that it is known \cite{Ellison:2021vth} that every twisted quantum double with Abelian anyons can be realized as the ground-state of a Pauli stabilizer code; our result implies that the qudit dimension of such a stabilizer code cannot be arbitrary.

We now provide a high-level argument for our results. Our first objective is to demonstrate that stabilizer states cannot arise in the sufficiently low-energy spectrum of a non-Abelian topological quantum phase. Since no parent Hamiltonian is specified for the stabilizer state, making sense of this statement requires extracting the anyon data directly from the wavefunction. There are currently two such approaches: an algebraic approach~\cite{Haah:2014iuw} and an entanglement-bootstrap approach~\cite{Shi:2019mlt}. In this work, we prove that stabilizer states that satisfy the entanglement bootstrap axioms can only realize a restricted set of Abelian topological orders. This statement is clearly true if the anyons are defined with respect to the parent Hamiltonian which is given by the sum of the stabilizer generators.  In that case, excitations are created by Pauli operators, which immediately rules out non-Abelian anyon statistics. A defining feature of non-Abelian anyons is a non-deterministic fusion:
\begin{equation}
a\times b=\sum_cN_{ab}^cc,
\end{equation}
while Abelian anyons have definite fusion outcome. If all anyons are Pauli strings, then putting one on top of another gives a single Pauli string rather than a superposition. Hence the fusion outcome is deterministic. Moreover, the fact that anyons are Pauli strings strongly constrains the possible mutual braiding phases. Consider a braiding experiment where an anyon crosses another. Since the strings are Pauli's, crossing the operators only produce a phase $e^{i\phi}$:
\begin{equation}\label{eq:braiding_experiment}
\vcenter{\hbox{\includegraphics[height=1.2cm]{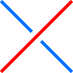}}}\ 
\quad=\quad
\vcenter{\hbox{\includegraphics[height=1.2cm]{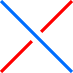}}}\quad e^{i\phi},
\end{equation}
where $\phi$ is an integer power of $2\pi i/q$, $q$ is qudit dimension. 

These conclusions can be made on a single disk with radius $O(d)$, when the state is acted upon by depth-$d$ local circuit. In total, we can find $O(n/d^2)$ such patches. In the non-Abelian ground state, the reduced density matrix on each patch cannot be a stabilizer projection state, namely states which are proportional to spectral projectors of a stabilizer group, see Fig~\ref{fig:patches_with_magic}.
\begin{figure}
    \centering
    \includegraphics[width=.9\linewidth]{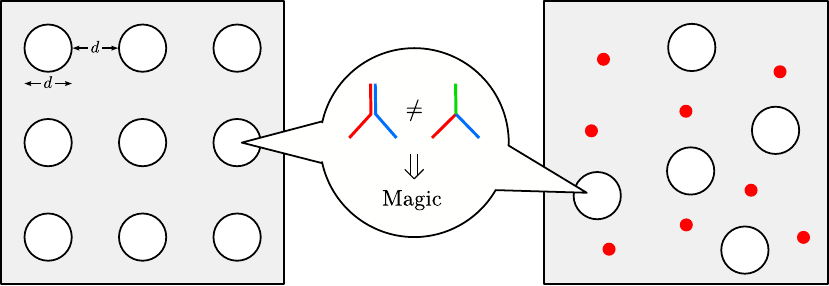}
    \caption{Patch construction for non-Abelian states. In a non-Abelian string-net ground-state (left) or a finite energy-density excited state (right) evolved by a depth-$d$ local unitary circuit, there are extensively-many well-separated patches that are anyon-free and hence locally agree with a ground state density matrix. The entanglement bootstrap axioms allow us to define and extract universal topological data within each such patch. A fictitious experiment, involving creating and braiding anyons in each patch reveals the existence of a non-trivial fusion space, which we demonstrate to be inconsistent with stabilizerness. This leads to extensive, long-range magic in either case.} %\YZ{Patch construction for non-Abelian states. Left: on a non-Abelian string-net ground state evolved by a depth-$d$ local unitary, one can find an extensive number of well-separated clean patches that are anyon free and hence agree with the ground state density matrix. Right: on a low-energy excited state with a low density of anyon excitations, the anyon endpoints are enlarged to regions of width $O(d)$ after the depth-$d$ circuit, but one can still find an extensive number of well-separated clean patches, as long as $d$ doesn't scale badly with the energy density. Zooming in on each patch,} }
    \label{fig:patches_with_magic}
\end{figure}
We use this property to show that there is at least $O(n/d^2)$ magic in total. This line of argument is similar to~\cite{Li:2025vru}, but here we need to show an exponentially small fidelity between the state over patches and a stabilizer state. A challenge is that the stabilizer state can be entangled across the patches, meaning that the fidelity isn't simply multiplicative. We overcome this by proving a multiplicative fidelity upper bound, using the notion of fidelity under stabilizer measurements.

The same physical picture also extends to low-energy states. A state with low energy has most of its weight on eigenstates with a low density of anyon excitations. After applying a depth-$d$ local unitary, each anyon endpoint is enlarged to a region of diameter $O(d)$. Provided the circuit depth is not too large compared with the typical separation between anyons, one can still find many well-separated clean, anyon-free patches whose reduced density matrices agree with those of the ground state. Utilizing the same local obstruction and the multiplicative fidelity upper bound on these clean patches yields an extensive lower bound on the long-range magic of low-energy states.

In the argument above we have shown that there is at least extensive long-range magic. We further show that the long-range magic after applying a circuit of depth $d$ is upper bounded by $O(n/d^2)$ by presenting a local circuit that decreases the magic down to $O(n/d^2)$ at depth $d$. The circuit is given by the entanglement renormalization circuit, which coarse grains the system and extracts the magic at increasingly IR length scales. String-net models are fixed points under entanglement renormalization, in the sense that the circuit is composed of exactly the same gates at every scale. Importantly, in the non-Abelian case, each gate is non-Clifford, suggesting that there is extensive magic at all scales.

\section{Stabilizer states only host restricted Abelian excitations}\label{sec:EB}

In this section, we show that stabilizer states evolved by constant-depth local unitaries cannot realize the ground state of any non-Abelian string-net model. In fact, we prove a stronger, local statement: if, after applying a depth-$d$ local unitary to a string-net ground state, there exists a subregion of width $\Theta(d)$ whose reduced state is a stabilizer projector state---a state proportional to spectral projectors of a stabilizer group, then the string-net model is necessarily Abelian. We further show that the corresponding anyonic excitation operators can be chosen to be Pauli operators, which forces all mutual braiding phases to be quantized by the qudit dimension.

The main idea behind the proof is to combine the structure of states that host non-Abelian anyons, phrased in the formalism of entanglement bootstrap~\cite{Shi:2019mlt}. If the underlying state hosts a non-Abelian anyon, there should be a nontrivial fusion space that can store quantum information. In particular, there should be states supported on the fusion space that do not commute with each other. Entanglement bootstrap lets us port this abstract statement about fusion space to a concrete structural property of a given physical state. In particular, this structure is shown to be incompatible with stabilizer states, which already rules out non-Abelian anyons for stabilizer states. We further prove constraints on the topological $S$-matrix of an anyon model based on stabilizer formalism, which provide a more fine-grained information.

% Before proceeding to the results, let us make a few remarks on an important subtlety. The original axioms that were advocated in entanglement bootstrap~\cite{Shi:2019mlt} are known to be fragile. They can break under constant-depth quantum circuits, making the framework not applicable to some topological states [cite spurious tee]. Nonetheless, we emphasize that some of the structural properties derived under these axioms do survive under constant-depth quantum circuits, a fact that we prove in this section. The prime examples include the (i) braiding statistics of the anyons and (ii) a set  of states known as the information convex set~\cite{Shi:2019mlt}. Our argument relies solely on the robustness of these objects, and as such, do not suffer from the examples in [cite spurious tee].

% We also remark that there is currently an ongoing work that aim to propose a more robust version of the axioms [cite], which is briefly described in~\cite{Kim:2024bsn}. We also describe an argument based on this alternative approach, for completeness.

\subsection{Entanglement bootstrap perspective on non-Abelian TQO}
\label{subsec:nonabelian_eb}

The ground state of a non-Abelian topological quantum order is a many-body wavefunction in which gapped  excitations can exhibit non-Abelian anyon statistics. Defining these excitations naturally requires knowledge of the Hamiltonian and its low-energy spectrum. In contrast, the entanglement bootstrap~\cite{Shi:2019mlt} provides a way to define these excitations from just the ground-state wavefunction, itself. In this approach, one posits that the given many-body ground-state  obeys certain locally-checkable conditions. Once these conditions are satisfied, topological excitations can be defined by investigating obstructions to constructing the state from knowledge of the state on local patches and in various topologies. % Universal properties of the topological quantum order are  %Physically, a the ground state of a non-Abelian topological quantum order is a many-body wavefunction in which gapped excitations can exhibit non-Abelian anyon statistics. While intuitive, this physical picture implicitly uses Hamiltonian and low-energy excited states, which are objects that may not be definable from a single state a priori. Entanglement bootstrap [22] is an approach that makes this possible. In this approach, one posits that the given ground state wavefunction obeys certain locally checkable conditions. Once these conditions are satisfied, the familiar notion of low-energy excitations and their properties can be defined just from that single state. Once these conditions are satisfied, the familiar notion of low-energy excitations and their properties can be defined just from that single state. 

The axioms of the entanglement bootstrap are summarized in Fig.~\ref{fig:EB_axioms_old}. More precisely, we consider a many-body quantum state $\sigma_0$ that obeys the following conditions.  Throughout our discussion here, we use the distance metric of the ambient space on which the qubits/qudits are placed (e.g., Euclidean distance). 
Let $\mu$ be a set of balls of radius $r_1$. For each $r_1$-ball, we choose a finite set of subsystems topologically equivalent to Fig.~\ref{fig:EB_axioms_old}, and demand that these axioms are satisfied. With an appropriate overlapping choice of this set, it follows that these conditions continue to hold at a larger scale, say $r'$ for any $r' \ge cr_1$ for some sufficiently large constant $c$. We envision choosing $r_1$ and the finite set so that this statement holds.\footnote{A rigorous formulation of this statement will be discussed in Ref.~\cite{EB:newpaper}. However, for string-nets~\cite{wen_string_net2005}, the fact that axioms hold at a larger scale follows from an easier argument; the axioms in Fig.~\ref{fig:EB_axioms_old} follow from the entanglement area law, which was shown in Ref.~\cite{Kitaev:2005dm,Levin_2006}.}

\begin{figure}
\centering
\subfigure[]{\includegraphics[width=0.32\textwidth]{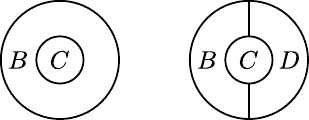}\label{fig:EB_axioms_old}}\qquad\qquad\qquad
\subfigure[]{\includegraphics[width=0.32\textwidth]{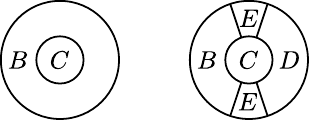}\label{fig:EB_axioms_new}}
\caption{(a) Original axioms of the entanglement bootstrap. Left: $S_{\sigma_0}(BC)+S_{\sigma_0}(C)-S_{\sigma_0}(B)=0$. Right: $S_{\sigma_0}(BC)+S_{\sigma_0}(CD)-S_{\sigma_0}(B)-S_{\sigma_0}(D)=0$. (b) New axioms, which are weaker but have the advantage of being robust against constant-depth quantum circuits~\cite{Kim:2024bsn}. Left: $S_\sigma(BC)+S_\sigma(C)-S_\sigma(B)=0$. Right: the state can be locally extended from $BE$ to $BC$. That is, let $A$ be the purification of $BCDE$, there exist a channel $\Phi_{BE\rightarrow BC}$ such that $\calI_A\otimes\Phi_{BE\rightarrow BC}(\sigma_{ABE})=\sigma_{ABC}$.}
\label{fig:EB_axioms}
\end{figure}

%{\color{purple} 

Using this approach, we seek to identify universal properties of topological ground states that support non-Abelian anyons. Specifically, we study states that are related by a constant-depth unitary circuit to a ground-state wavefunction satisfying the entanglement bootstrap axioms~\cite{Shi:2019mlt}. Canonical examples include string-net ground states~\cite{wen_string_net2005} and states obtained from them by the action of a constant-depth circuit.

A subtlety is that the entanglement bootstrap axioms themselves are not necessarily preserved under such constant-depth circuits. Consequently, even if a topologically ordered ground state lies in the same phase as a reference state that does satisfy the axioms, such as a string-net ground state, the axioms do not automatically apply directly to that state~\cite{Kitaev:2005dm,Levin_2006,bravyi2008unpublished,Zou:2016dck}. Nevertheless, we show that certain consequences of the axioms remain stable under constant-depth circuits. In particular, these robust consequences constrain the structure of the information convex set, defined below, and encode universal topological data associated with the non-Abelian character of the low-energy excitations. 

For the string-net models considered in this work, we use the identification between the data extracted from a given state via the bootstrap axioms, such as the anyon labels and quantum dimensions, with the intrinsic topological data of the phase represented by that wavefunction. Indeed, a recent work established such an identifications for string-net wavefunctions~\cite{bols2025sector,bols2026sector} as we describe in subsequent sections.
%\IK{Indeed, a recent work established such a connection for string-net wavefunctions~\cite{bols2025sector,bols2026sector}; see ~\cite[Remark 4.8]{bols2025sector} in particular.} 

%}

% We also briefly discuss an argument based on a complementary approach, in which one begins with an alternative, robust set of axioms~\cite{Kim:2024bsn} [cite new paper].

We now define the \emph{information convex set}. Let $\Om$ be a region and $\Om^{+}$ be a $r_1$-neighborhood of $\Om$. 
\begin{definition}[Information convex set~\cite{Shi:2019mlt}]
The information convex set of $\sigma_0$ on $\Omega$ is 
\begin{equation}
    \Sigma(\Omega)= \left\{ \rho_{\Omega}\,\big|\, \rho_{\Omega} = \mathrm{Tr}_{\Omega^+ \setminus \Omega}( \rho_{\Omega^+}), \,\rho_{\Omega^+} \in \tilde{\Sigma}(\Omega)\right\},
\end{equation}
where
\begin{equation}
    \tilde{\Sigma}(\Omega) = \left\{ \rho_{\Omega^+}\,\big|\,\mathrm{Tr}_{\Omega^+\setminus b} (\rho_{\Omega^+}) =\sigma_{0b}, \forall \,b\in \mu\right\},
\end{equation}
\end{definition}
\noindent 
Here, $\sigma_{0b}$ is the reduced density matrix of $\sigma_0$ over the ball $b$. Colloquially, the information convex set $\Sigma(\Omega)$ is a set of states which are locally identical to the reference state on $\Omega$. 

If the underlying state $\sigma$ satisfies the axioms in Fig.~\ref{fig:EB_axioms_old}, for any region $\Omega$ that is sufficiently ``nice'', the information convex set obeys a structure theorem~\cite{Shi:2019mlt,EB:newpaper}. More precisely, we say $\Omega$ is $r$-\emph{regular} if its boundary curvature is at most $r^{-1}$ and it admits an inner $r$-collar, a neighborhood of the boundary inside the region, and of thickness $r$. The precise statement is that for $\Omega$ that is $r$-regular for $r\gg r_1$, its information convex set obeys certain structure theorems~\cite{Shi:2019mlt}. 

For instance, if $\Omega$ is an annulus, its information convex set is isomorphic to a convex set of density matrices whose extreme points are orthogonal to each other. In particular, each extreme point corresponds to the topological charge of the underlying anyon theory. If $\Omega$ is an $n$-holed disk for $n\ge 2$, $\Sigma(\Omega)$ is isomorphic to a direct sum of state spaces over finite-dimensional Hilbert spaces. Moreover, if we have two regular regions $\Omega$ and $\Omega'$ that are isotopic to each other, there is a bijective map realized by a reversible quantum channel between $\Sigma(\Omega)$ and $\Sigma(\Omega')$; this is known as the isomorphism theorem~\cite{Shi:2019mlt}.

Importantly, these structure theorems continue to hold even if the underlying reference state is changed to $\sigma=U\sigma_0 U^{\dagger}$, where $U$ is a constant-depth quantum circuit. We provide a sketch of the argument below.

\begin{figure}
    \centering
    \includegraphics[width=0.5\linewidth]{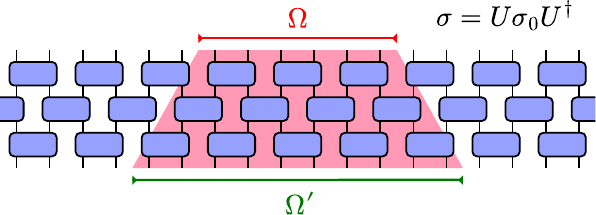}
    \caption{The state $\sigma = U\sigma_{0}U^{\dagger}$ obtained by a constant-depth circuit, along with regions $\Omega \subset\Omega'$ used in the proof of Proposition \ref{prop:info_convex_u_no_u}. Here, the image is shown in one spatial dimension for convenience. }
    \label{fig:Prop1_fig}
\end{figure}

\begin{proposition}
\label{prop:info_convex_u_no_u}
    Let $\Sigma'(\Omega)$ be an information convex set defined with respect to $\sigma=U\sigma_0 U^{\dagger}$, where $U$ is a depth-$d$ local quantum circuit. If $\Omega$ is $r$-regular for $r\gg d$, 
    \begin{equation}
        \Sigma(\Omega) \cong \Sigma'(\Omega).
    \end{equation}
    In particular, the bijection can be realized by a reversible quantum channel. 
\end{proposition}

\begin{proof}
    We first construct a channel from $\Sigma(\Omega)$ to $\Sigma'(\Omega)$. We consider $\Omega'\supset \Omega$ that is isotopic to $\Omega$ and includes the past lightcone of $\Omega$. By the isomorphism theorem~\cite{Shi:2019mlt}, there is a channel that maps $\Sigma(\Omega)$ to $\Sigma(\Omega')$. We now apply the portion of the circuit $U$ whose support can affect $\Omega$, and then trace out $\Omega'\setminus \Omega$. The resulting states belong to $\Sigma'(\Omega)$. An example of these regions, and the relevant lightcone are shown in Fig. \ref{fig:Prop1_fig}.

    In the opposite direction, we consider a $\Omega'' \subset \Omega$ isotopic to $\Omega$ such that the future lightcone of $\Omega''$ is included in $\Omega$. Starting from an element of $\Sigma'(\Omega'')$. we can apply the the inverse of the lightcone of $\Omega''$ and trace out $\Omega \setminus \Omega''$, which yields elements in $\Sigma(\Omega'')$. Since $\Omega''$ is isotopic to $\Omega$, we can apply the isomorphism theorem to obtain an element in $\Sigma(\Omega)$. 
\end{proof}

Here it is important to note that we used the isomorphism theorem --- a property that follows from the entanglement bootstrap axioms [Fig.~\ref{fig:EB_axioms_old}] --- only with respect to the original state $\sigma_0$. Therefore, even though $\sigma=U\sigma_0 U^{\dagger}$ may not satisfy the axioms, there is a bijective map between $\Sigma'(\Omega)$ and $\Sigma(\Omega)$. In particular, if $\Omega$ is a $n$-holed disk with $n\ge 2$, $\Sigma'(\Omega)$ is isomorphic to a direct sum of state spaces of some finite-dimensional Hilbert spaces. We also note that, as a simple corollary of Proposition~\ref{prop:info_convex_u_no_u}, the isomorphism theorem also holds for $U\sigma_0 U^{\dagger}$.

The non-Abelian character of the anyon excitations is intimately related to the information convex sets of $(n\geq 2)$-holed disk. For such subregions, the information convex set can be further decomposed into a direct sum over the topological charges measured on thin annular neighborhoods of its boundary components, namely the $n$ inner boundaries and the outer boundary, see Fig.~\ref{fig:EB_punctured_disk}. Once these topological charges are fixed, we are left with a state space of a finite-dimensional Hilbert space, whose dimension is precisely the dimension of the fusion space. In particular, the existence of a non-Abelian anyon necessarily implies that there should be a nontrivial fusion space. 
\begin{figure}
    \centering
    \includegraphics[width=0.32\linewidth]{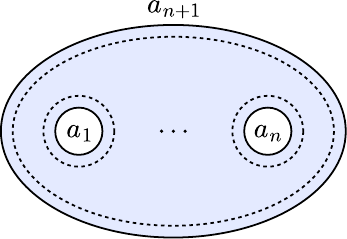}
    \caption{The $n$-holed disk. The information covex set of the $n$-holed disk decomposes into a direct sum over convex sets labeled by topological charges measured on the thin annular neighborhoods of its boundary components. $a_1,\ldots,a_n$ are the charges in the holes measured by thin annuli surrounding them; $a_{n+1}$ is the overall charge measured by the thin annular neighborhood of the outer boundary. This particular sector is isomorphic to the fusion space where $a_1,\ldots,a_n$ fuse into $a_{n+1}$.}
    \label{fig:EB_punctured_disk}
\end{figure}

We can formalize this intuition as follows. For an $n$-holed disk $Z$, its information convex set can be decomposed into mutually orthogonal sectors as 
\begin{equation}
    \Sigma(Z) = \bigoplus_{a_1, \ldots, a_{n+1}} \Sigma_{a_1\ldots a_n}^{a_{n+1}}(Z),
\end{equation}
where $a_1, \ldots, a_{n+1}$ label topological charges. Each $\Sigma_{a_1\ldots a_n}^{a_{n+1}}(Z)$ is isomorphic to a state space of some finite-dimensional Hilbert space, called as the \emph{fusion space}. Their dimension is denoted as $N_{a_1\ldots a_n}^{a_{n+1}}$. It was shown in Ref.~\cite{Shi:2019mlt} that these objects satisfy the fusion rules of the anyon theory in Ref.~\cite{Kitaev:2005hzj}.

For the string-net wavefunctions, it is possible to make the connection between information convex sets and the anyon data explicit. In Ref.~\cite{bols2025sector}, the authors demonstrated that the information convex sets of annulus and two-hole disk for the string-net wavefunction contains information about anyon types and fusion rules; see Lemma 4.7 and Remark 4.8 therein. Since ~\cite[Lemma 4.7]{bols2025sector} applies to arbitrary regions, it should also apply to more general regions considered below.

Using this connection, we show below that for $n\geq 3$-holed disk, there should be a nontrivial fusion space. Consider the process of fusing three anyons $a$, $\bar a$, and $a$ to total charge $a$. Let $N_{a\bar a a}^a$ denote the dimension of the fusion space, i.e. the number of internal states in the fusion tree
\begin{equation*}
\includegraphics[height=3cm]{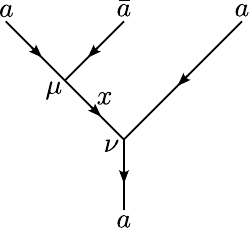}.
\end{equation*}
In this basis, one first fuses $a$ and $\bar a$ to an intermediate anyon $x$, and then fuses $x$ with the remaining $a$ back to $a$. For each allowed $x$, let $\mu\in\{1,\dots,N_{a\bar a}^x\}$ label the fusion channel at the first vertex and let $\nu\in\{1,\dots,N_{x a}^a\}$ label the fusion channel at the second vertex, where $N_{x a}^a = N_{a\bar a}^x$. The multiplicity simply counts the number of possible combinations of $x$, $\mu$ and $\nu$, given by
\begin{equation}
N_{a\bar a a}^a=\sum_x (N_{a\bar a}^x)^2=1+\sum_{x\neq 1}(N_{a\bar a}^x)^2,
\end{equation}
where we used $N_{a\bar a}^1=1$. If an anyon $a$ is non-Abelian, it has a quantum dimension $d_a>1.$ Since~\cite{Kitaev:2005hzj}
\begin{equation}
d_a^2=d_a d_{\bar a}=\sum_x N_{a\bar a}^x\, d_x,
\end{equation}
the condition $d_a>1$ implies that $a\times\bar a$ must have at least one non-vacuum fusion outcome. In other words, there exists some $x\neq 1$ such that $N_{a\bar a}^x>0$, and hence
\begin{equation}\label{eq:higher_fusion_multiplicity}
N_{a\bar a a}^a>1.
\end{equation}

Eq.~\eqref{eq:higher_fusion_multiplicity} implies that the information convex set of a three-holed disk $Z$ must be nontrivial. In particular, there should be a topological charge $a$ such that $\Sigma_{a \bar{a} a}^a(Z)$ is isomorphic to a state space of some nontrivial finite-dimensional Hilbert space. For the states in that abstract space, it is a simple fact that there exists at least a pair of states that do not commute with each other. From this, we show that there should exist a pair of states in the information convex set $\Sigma(Z)$ that do not commute with each other.

\begin{lemma}
\label{lemma:commutation_bijection}
    Let $\mathcal{D}(\mathcal{H})$ be a set of density matrices of a finite-dimensional Hilbert space $\mathcal{H}$ with $\dim\calH>1$. If $\mathcal{D}(\mathcal{H})$ is isomorphic to a convex set $\mathrm{S}$ of finite-dimensional density matrices by a reversible channel, there exists $\rho, \rho' \in \mathrm{S}$ such that $[\rho, \rho']\neq 0$.
\end{lemma}
\begin{proof}
    We prove the contrapositive. Let $\Phi:\mathcal{D}(\mathcal{H}) \to \mathrm{S}$ be a reversible channel. If $[\rho, \rho']=0$ for all $\rho, \rho'\in \mathrm{S}$, all the elements of $\mathrm{S}$ can be diagonalized simultaneously. By linearity, the action of $\Phi$ can be written as 
    \begin{equation}
        \Phi(\sigma) = \sum_i \text{Tr}(\sigma M_i) |i{\rangle\langle i}|
    \end{equation}
    for a set of operators $\{ M_i\}$, where $\{ |i\rangle\}$  is an orthonormal basis set  in the support of $\mathrm{S}$. In particular, because $\Phi$ is CPTP, it follows that $\{M_i\}$ forms a POVM.

    Note $\Phi$ is entanglement-breaking. That is, consider an extension $\mathcal{I}_A\otimes \Phi$. The image of this channel consists of separable states with respect to the bipartition between $A$ and the rest. Thus this channel cannot be reversible, which is a contradiction.  
\end{proof}

From Lemma~\ref{lemma:commutation_bijection}, it follows that there should be a pair of states in $\Sigma_{a\bar{a}a}^a(Z)$ that do not commute with each other, for some topological charge $a$. This is a key property of the ground state wavefunction of a non-Abelian topological order that holds true universally. As we show in Section~\ref{subsec:stab_projection}, stabilizer states are fundamentally incompatible with this constraint. We summarize this into the following theorem.
\begin{theorem}
\label{thm:nonabelian_noncommuting}
    Let $\sigma$ be a state obtained by evolving a non-Abelian string-net ground state with a depth-$d$ local unitary. Let $\Omega$ be a $r$-regular three-holed disk. The information convex set of $\Omega$ defined with respect to $\sigma$ contains a pair of non-commuting states, provided that $r\ge c d$ for some sufficiently large constant $c$.
\end{theorem}

\noindent 
\textbf{Note}: Theorem~\ref{thm:nonabelian_noncommuting} can be proved under a weaker assumption, starting with a recently studied robust axioms of entanglement bootstrap~\cite{Kim:2024bsn,EB:newpaper}; see Fig.~\ref{fig:EB_axioms_new}. Our argument can be viewed as a poor-man's version of that statement. We take this approach for two reasons. First, our argument hinges only on known results in the literature. Second, the concrete examples we envision are string-net wavefunctions, for which our argument does apply.

\subsection{Trivial fusion space in stabilizer projection states}
%\subsection{Stabilizer projection state: triviality of fusion space}
\label{subsec:stab_projection}

In Section~\ref{subsec:nonabelian_eb}, we proved a robust property of the ground states of non-Abelian topological orders. There should be a pair of states in the information convex set of a three-hole disk that do not commute with each other. This statement holds for any state of the form $\sigma=U\sigma_0 U^{\dagger}$, where $\sigma_0$ satisfies the entanglement bootstrap axioms [Fig.~\ref{fig:EB_axioms_old}] and $U$ is a constant-depth quantum circuit; see Theorem~\ref{thm:nonabelian_noncommuting}. 

Now we aim to show that this structure is incompatible with any stabilizer state. Specifically, we posit that on a ball region with sufficiently large $O(d)$ radius, a stabilizer projection state (Definition~\ref{def:stabilizer_projection_states}) arises as the reduced density matrix of a non-Abelian ground state evolved by a depth-$d$ local unitary.  This will lead to a contradiction, implying that such ground states have long-range magic.

Let us first set up our notations. On a $q$-dimensional qudit, a local basis can be labeled by $\{\ket{j}\}_{j=0}^{q-1}$. The Pauli $X$ and $Z$ matrices are defined as clock and shift operators, generalizing the qubit case:
\begin{equation}\begin{aligned}
Z\ket{j}&=\omega^j\ket{j},\qquad\omega:=e^{2\pi i/q}, \\
X\ket{j}&=\ket{j+1\ \mod q}.
\end{aligned}\end{equation}
The multi-qudit Pauli group is generated by the tensor product of single qubit Pauli operators and $\omega I$. A pure stabilizer state is the unique common eigenstate of a group of commuting Pauli operators. More generally, we can define a stabilizer projection state, defined below:
\begin{definition}[Stabilizer projection states]\label{def:stabilizer_projection_states}
Given a stabilizer group $S$ with generators $\calG(S)$, the associated stabilizer projection state is the maximally mixed state in the stabilizer code associated with $S$. That is, the state is proportional to the projector
\begin{equation}
\prod_{g\in\calG(S)}P(g),
\end{equation}
where $P(g)$ denotes the projector onto the eigenspace of $g$ with eigenvalue $1$. If $\delta$ is the order of $g$, i.e. the smallest positive integer such that $g^\delta\propto I$, then $P(g)=\delta^{-1}(1+g+\cdots+g^{\delta-1})$.
\end{definition}
\noindent
The main motivation is that the reduced density matrix of pure stabilizer states take the form of stabilizer projection states.

If the stabilizer projection state $\sigma$ (which we refer to as the reference state) is obtained by applying a depth-$d$ local unitary $U$ to a state satisfying the entanglement bootstrap axioms, it satisfies a few nontrivial properties. While the axioms in Fig.~\ref{fig:EB_axioms_old} may not continue to hold, a weaker set of axioms in Fig.~\ref{fig:EB_axioms_new}~\cite{Kim:2024bsn} are satisfied. These will be used repeatedly, so we discuss them in more detail below.
\begin{itemize}
\item \emph{Decoupling}: $I_\sigma(A:C)=S_\sigma(BC)+S_\sigma(C)-S_\sigma(B)=0$ in the configuration below. This implies there is zero correlation between $A$ and $C$.
\begin{equation*}
\centering\includegraphics[width=0.18\textwidth]{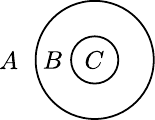}
\end{equation*}

\item \emph{Local extendibility}~\cite{Kim:2024bsn}: Given the reduced state $\sigma_{ABE}$ in the configuration shown below
\begin{equation*}
\centering\includegraphics[width=0.18\textwidth]{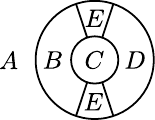}
\end{equation*}
there exists a channel $\Phi_{BE\rightarrow BC}$ such that $\calI_A\otimes\Phi_{BE\rightarrow BC}(\sigma_{ABE})=\sigma_{ABC}$.  Crucially, this extension channel is determined by $\sigma_{BCE}$, and therefore depends only on the state in the local neighborhood of the region being extended.
\end{itemize}
The partitions described above pertain to two spatial dimensions. Generalizations to higher dimensions are straightforward. For instance, in three dimensions, one  can rotate the configurations about a horizontal axis.

We will now use these properties to investigate the reduced state of $\sigma$ over a large, disk-like region $B$.  We will show that $\sigma_{B}$ has a locally generated stabilizer group. To demonstrate this, we first observe that the state $\sigma_B$ over a disk-like region $B$ is completely fixed by the reduced density matrices over small balls that cover a slightly enlarged region~\cite{Kim:2024bsn}. This is a consequence of the fact that for a ball-like region, the information convex set has only one element. To see why, Let $B^+$ be the thickening of $B$---the $r_1$-neighborhood of $B$, and let $\mu$ be the set of all balls $b$ with radius $\Theta(d)$ that are supported on $B^+$. By sequentially applying the local extending channels, one can reconstruct $\sigma_B$ from the collection of local marginals $\{\sigma_b\}_{b\in\mu}$. Importantly, each local extending channel is determined by the density matrix on its input and output neighborhoods~\cite{Kim:2024bsn}.  Thus, the extending channel $\Phi$ which takes the state on a given small ball and produces the state on $B$, i.e. $\Phi(\sigma_{b}) = \sigma_{B}$, is entirely fixed by this collection of local marginals.

\begin{table}[b]
\centering
\begin{tabular}{ll}
\hline
Notation & Definition \\
\hline
$\calG(S)$ & generators of the stabilizer group $S$ \\
$S_r(A)$ & the stabilizers of the reference state that are fully supported on $A$ \\
$S_A$ & the stabilizer group generated by local stabilizers in $A$ \\
$S_A(B)$ & the stabilizers in $S_A$ that are fully supported on $B$ \\
$L_A$ & the Pauli operators on $A$ that commute with $S_A$ \\
$L_A(B)$ & the operators in $L_A$ that are fully supported on $B$ \\
\hline
\end{tabular}
\caption{List of notation used in this section.}
\label{tab:notations_stabilizers}
\end{table}

An immediate consequence is that $\sigma_{B}$ has a locally-generated stabilizer group. Consider the stabilizer group on each ball $b\in\mu$, which we denote as $S_r(b)$; these are the set of stabilizers of $\sigma_{B^{+}}$ which are exclusively supported on $b$.  We define $S_{B^+}$ as the stabilizer group generated by the collection of all these stabilizers $\cup_{b\in\mu}S_r(b)$. The notation used in this
construction is summarized in Table~\ref{tab:notations_stabilizers}. We then consider the corresponding stabilizer projection state:
\begin{equation}
\pi_{B^+}\propto \prod_{g\in\calG(S_{B^+})}P(g).
\end{equation}
The state $\pi_{B^+}$ manifestly reproduces the correct reduced density matrices on all balls $\pi_{b} = \sigma_{b}$ for all $b\in \mu$. As a consequence, the local extending channel $\Phi$ that extends $\sigma_b$ to $\sigma_B$ should work for $\pi_{B^+}$ as well\cite{Kim:2024bsn}, that is, $\Phi(\pi_b)=\pi_B$. With $\sigma_b=\pi_b$, we are led to the conclusion that 
%$\pi_{B}$ is also a locally-extendable state, and can be constructed by applying local extending channels to the state of a fixed ball $\Phi'(\pi_{b}) = \pi_{B}$.  Furthermore, since $\Phi'$ is entirely fixed by the marginals on balls $b\in \mu$ for which $\pi_{b} = \sigma_{b}$, the extending map is identical to that for the state $\sigma$, i.e. $\Phi' = \Phi$.  Then we are led to the conclusion that 
\begin{equation}
\sigma_B=\Tr_{B^+\backslash B}\pi_{B^+}.
\end{equation}
%In our context, $\sigma_{B^+}$ is a stabilizer projection state. This further implies that the stabilizer group of a ball-like region can be generated by local stabilizers, as shown in Fig.~\ref{fig:EB_locally_generated_stabilizers}. On each ball $b\in\mu$, we find the stabilizer group $S_r(b)$, defined as the set of stabilizers of $\sigma_{B^+}$ that are only supported on $b$. Define a stabilizer group $S_{B^+}$ generated by the collection of all these stabilizers $\cup_{b\in\mu}S_r(b)$. Consider the stabilizer projection state associated with $S_{B^+}$:
%\begin{equation}
%\pi_{B^+}\propto \prod_{g\in\calG(S_{B^+})}P(g).
%\end{equation}
To summarize, we have the following:
\begin{proposition}[Stabilizer of balls can be locally generated]
Let $S_r(B)$ be the stabilizers of reference state $\sigma$ that are fully supported on a ball $B$. Let $S_{B^+}$ be the stabilizer group generated by all stabilizers of $\sigma$ supported on balls $b\in\mu$, and let $S_{B^+}(B)$ be the elements of $S_{B^+}$ that are fully supported on $B$. Then
\begin{equation}\label{eq:EB_locally_generated_stabilizers}
S_r(B)=S_{B^+}(B)\subseteq S_{B^+}.
\end{equation}
\end{proposition}
%The state $\pi_{B^+}$ manifestly reproduces the correct reduced density matrices on all balls. For any $b\in\mu$, the reduced state of the reference state on $b$ is precisely the stabilizer projection state associated with $S_r(b)$. The same is true for $\pi_{B^+}$: by construction it is stabilized by $S_r(b)$ on $b$, and no additional stabilizers supported on $b$ are imposed. Therefore $\pi_{B^+}$ and $\sigma_{B^+}$ have the same reduced density matrix on every ball in $\mu$. Since the reduced state on $B$ is uniquely determined by the collection of reduced density matrices on these balls, it follows that $\pi_{B^+}$ and $\sigma_{B^+}$ must agree on $B$:
%\begin{equation}
%\sigma_B=\Tr_{B^+\backslash B}\pi_{B^+}.
%\end{equation}
%As a result, we have the following:
%\begin{proposition}[Stabilizer of balls can be locally generated]
%Let $S_r(B)$ be the stabilizers of reference state $\sigma$ that are fully supported on $B$. Let $S_{B^+}$ be the %stabilizer group generated by all stabilizers of $\sigma$ supported on balls $b\in\mu$. Then
%\begin{equation}\label{eq:EB_locally_generated_stabilizers}
%S_r(B)\subseteq S_{B^+}.
%\end{equation}
%\end{proposition}
\begin{figure}
    \centering
    \includegraphics[width=0.7\linewidth]{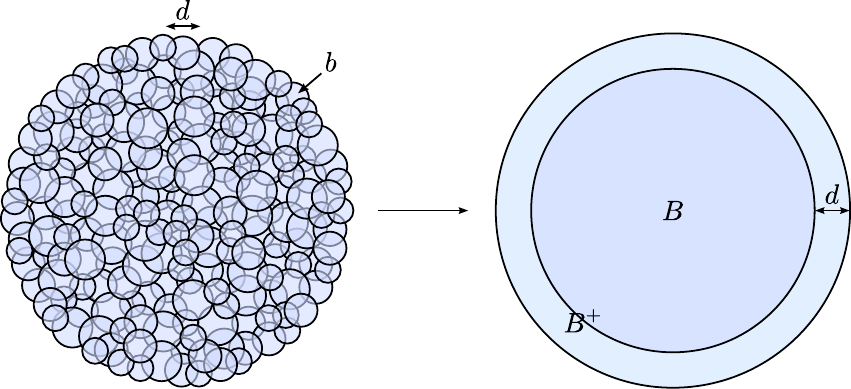}
    \caption{On a stabilizer projection state satisfying local extendibility, the stabilizers supported on a ball $B$ can be generated by the stabilizers on small balls $b$ supported on the enlarged ball $B^+$, see Eq~\eqref{eq:EB_locally_generated_stabilizers}.}
    \label{fig:EB_locally_generated_stabilizers}
\end{figure}

Now take the reduced density matrix of the reference state over $\Om^+$, denoted as  $\sigma_{\Om^+}$. 
%{\color{red} IK: Any reason for doing this? Naively, one would choose the reference state to be a stabilizer projection state over a disk that includes $\Omega^+$. }\YZ{I agree. At the beginning of this section, we should specify that we are working on a stabilizer projection state on a large enough ball.}  
Let $S_{\Om^+}$ be the local stabilizer generators of $\sigma_{\Om^+}$, i.e. the stabilizer group generated by stabilizers supported on local balls in $\Om^+$.  Any state $\rho_{\Om^+}\in\tilde{\Sigma}(\Om)$ is manifestly stabilized by these local generators, $g\rho_{\Om^+} = \rho_{\Om^+}$ $\forall$ $g\in S_{\Om^+}$, since $\rho_{\Om^+}$ agrees with the reference state on every ball. Thus $\rho_{\Om^+}$ lies in the codespace of the code defined by $S_{\Om^+}$, and can thus be expanded as follows:
\begin{equation}\label{eq:expansion}
\rho_{\Om^+}=\Pi_{S_{\Om^+}}\sum_{\ell\in L_{\Om^+}}c_{\ell}\ell,\qquad \Pi_{S_{\Om^+}}=\frac{1}{|S_{\Om^+}|}\sum_{s\in S_{\Om^+}}s.
\end{equation}
Here, $L_{\Om^+}$ is the set of all Pauli operators supported on $\Om^+$ that commute with $S_{\Om^+}$, and $\Pi_{S_{\Om^+}}$ is the projector on the code space. We refer to $\ell \in L_{\Om^+}$ as logical operators for the remainder of this section.  We note that for a fixed state $\rho_{\Om^+}$, the expansion in terms of logical operators in Eq. (\ref{eq:expansion}) is of course not unique, since different logical operators could be related by multiplying with a stabilizer $g\in S_{\Om^{+}}$. 

To get the states in $\Sigma(\Om)$, we need to trace over the region $\partial \Om:=\Om^+\backslash \Om$. Before doing that, let $S_{\Om^+}(\Om)$ be the stabilizers in $S_{\Om^+}$ that are only supported on $\Om$ and let $\Pi_{S_{\Om^+}(\Om)}$ be the corresponding projector. Now $\Pi_{S_{\Om^+}}$ is equal to $\Pi_{S_{\Om^+}(\Om)}$ multiplied with some projectors that has support on $\partial\Om$. We can absorb these boundary projectors into the sum $\ell\in L_{\Om^+}$ and write
\begin{equation}
\rho_{\Om^+}=\Pi_{S_{\Om^+}(\Om)}\sum_{\ell\in L_{\Om^+}}c_{\ell}\ell,\qquad \Pi_{S_{\Om^+}(\Om)}=\frac{1}{|S_{\Om^+}(\Om)|}\sum_{s\in S_{\Om^+}(\Om)}s.
\end{equation}
When we trace out the boundary $\partial\Om$, the terms that survive are those $s\ell$ that are supported entirely on $\Om$. Such an operator still belongs to $L_{\Om^+}$, so we can rewrite
\begin{equation}
\rho_\Om=\Tr_{\Om^+\backslash \Om}(\rho_{\Om^+})=\Pi_{S_{\Om^+}(\Om)}\sum_{\ell\in L_{\Om^+}(\Om)}c_{\ell}\ell,
\end{equation}
where $L_{\Om^+}(\Om)$ is the set of $\ell\in L_{\Omega^+}$ with $\supp(\ell)\subseteq\Omega$. The notations are summarized in Table~\ref{tab:notations_stabilizers}.

What $\ell$'s can appear in the above sum? Intuitively, the sum of $\ell$'s encodes the logical information on $\Om$. For example, on an annular region in the toric code, we expect them to be non-contractible $X$ loops and $Z$ loops. The contributing terms should mutually commute, since the loops can always be deformed away from each other. In fact, these loops are actually the stabilizers of the toric code, once we consider all stabilizers on a disk region that contains the annulus---although they cannot be locally generated on the annulus, they can be locally generated on the disk.

To make this clear, consider the reference state on the ball $B$ that contains $\Om$. Let $S_{B^+}$ be the stabilizers of $\sigma$ on the enlarged ball $B^+$, and $L_{B^+}$ be the Pauli's on $B$ that commute with $S_{B^+}$. Any $\ell\in L_{\Om^+}(\Om)$ not only commutes with generators on balls inside $\Om^+$, but also commutes with generators on balls outside $\Om$ because of non-overlapping support. This leads to the identification
\begin{equation}
L_{\Om^+}(\Om)=L_{B^+}(\Om)\subseteq L_{B^{+}}(B).
\end{equation}

We will proceed to show that $\ell\in S_{r}(\Omega)$, so that there are no non-trivial logical operators on the ball. The intuition is that the reduced state $\sigma_{B}$ can be purified by a thin region of width $\Theta(d)$ surrounding $B$, due to the local entanglement structure of the state (the local decoupling property). Thus $\ell$ commutes with all stabilizers of a pure state on $B^{+}$, and must be an element of $S_{r}(B^{+})$; the restricted support of $\ell$ further implies $\ell\in S_{r}(\Omega)$. %Intuitively, the fact that $\supp(\ell)\subseteq \Omega$, together with the decoupling condition forbids logicals which are not part of the stabilizer group $\ell\notin S_{r}(B^+)$, as such a logical would necessarily generate long-range-entanglement. We now prove these statements. 
%On the ball, the intuition is that there is no non-trivial element in $L_{B^+}$ supported on $B$ other than the stabilizers themselves---the stabilizers are ``dense'' enough to make sure that the maximally mixed state in the code space is only entangled along the boundary. Roughly speaking, if there is some $\ell\in L_{B^+}(B)$ but $\ell\notin S_{B^+}(B)$, the reference state on $B^+$ must contain a mixed degree of freedom supported on $B$. This degree of freedom cannot be purified with $\partial B$, otherwise $\ell$ wouldn't commute with $S_{B^+}$\footnote{For example, take an EPR pair stabilized by $Z\otimes Z$ and $X\otimes X$. Any Pauli operator acting on a single qubit would fail to commute with at least one of $Z\otimes Z$ and $X\otimes X$.}. However, this is not possible due to the fact that the system is short-range entangled: the mutual information between $B$ and the region outside $B^+$ must vanish. The precise statement is the following:
\begin{lemma}[No non-trivial logicals on balls]
\begin{equation}\label{eq:EB_trivial_logicals}
L_{B^+}(B)\subseteq \tilde{S}_r(B),
\end{equation}
where $\tilde{S}_r(B)$ is the phase-extended stabilizer group generated by $S_r(B)$ and global phases $\omega=\exp(2\pi i/q)$, $q$ is qudit dimension.
\end{lemma}
\begin{proof}
Let $\ket{\Phi}$ be a purification of $B^+$ obtained by an auxiliary subsystem $A$, as shown in Fig.~\ref{fig:EB_trivial_logicals}. The $I(B:A)=0$ condition implies that the reduced state on $B$ and $A$ factorizes: $\rho_{BA}=\rho_B\otimes\rho_A$. Hence one can purify $\rho_B$ and $\rho_A$ separately into the state $\ket{\phi}_{B\partial B_1}\otimes\ket{\psi}_{\partial B_2A}$ and set $\partial B=\partial B_1\otimes \partial B_2$. Since $\ket{\Phi}$ is also a purification of $\rho_{BA}$, there must be a unitary $U_{\partial B}$ acting on the purifying space $\partial B$ such that $U_{\partial B}\ket{\Phi}=\ket{\phi}_{B\partial B_1}\otimes\ket{\psi}_{\partial B_2A}$. Under this unitary, the stabilizer group $S_r(B^+)$ is mapped to another group $S'_r(B^+)$.

Now take any $\ell\in L_{B^+}(B)$. $\ell$ must commute with all the stabilizers of the reference state, because $\ell$ commutes with $S_{B^+}$ and $S_r(B)\subseteq S_{B^+}$. Since $\ell$ is supported on $B$, it is unchanged by conjugation with $U_{\partial B}$. Hence, after applying $U_{\partial B}$ to the purification, $\ell$ still commutes with the transformed group $S'_r(B^+)$. In particular $\ell$ commutes with all the stabilizers of $\ket{\phi}_{B\partial B_1}$. For a pure state, any Pauli operator that commutes with its stabilizer group must itself be a stabilizer, up to a phase. Therefore $\ell$ is a stabilizer of $\ket{\phi}_{B\partial B_1}$ up to a phase. Moreover, since $\ell$ is supported on $B$, $\ell$ must be a stabilizer of the reference state on $B$ as claimed.
\end{proof}
\begin{figure}
    \centering
    \includegraphics[width=0.4\linewidth]{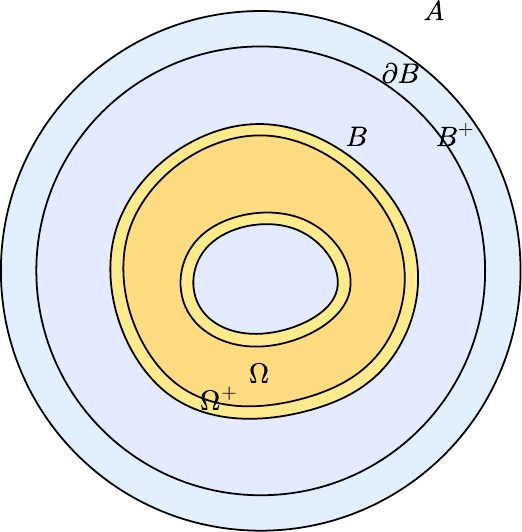}
    \caption{To show Eq.~\eqref{eq:EB_trivial_logicals}, we consider a ball $B$ that contains $\Omega^+$. Since $I(A:B)=0$, the inner ball can be disentangled with the outside region by applying a unitary on $\partial B$.}
    \label{fig:EB_trivial_logicals}
\end{figure}

The above lemma allows us to restrict the sum of $\ell$'s to the stabilizers of the reference state that are supported on $\Om$.
\begin{equation}\label{eq:EB_expansion_of_density_matrix}
\rho_\Om=\Pi_{S_{\Om^+}(\Om)}\sum_{\ell\in S_r(\Om)}c_{\ell}\ell.
\end{equation}
As a consequence, we find
\begin{equation}\label{eq:EB_trivial_fusion}
\wideboxed{
[\rho_\Om,\rho'_\Om]=0, \text{ for }\forall\rho_\Om,\rho'_\Om\in\Sigma(\Om)
}.
\end{equation}
We note that the argument up to Eq.~\eqref{eq:EB_trivial_fusion} is not specific to two spatial dimensions. 

\begin{theorem}\label{thm:EB_stabilizer_implies_Abelian}
Consider a non-Abelian string-net ground state evolved by a local unitary of depth $d$. Then a ball-like subsystem with diameter $O(d)$ cannot be realized by a stabilizer projection state.
\end{theorem}

\begin{proof}
Assume for contradiction that some ball-like subsystem of sufficiently large $O(d)$ diameter is a stabilizer projection state. Choose a three-holed disk $Z$ contained in that ball and with the same $O(d)$ thickness. The argument above then applies to $\Sigma(Z)$ and implies that all states in $\Sigma(Z)$ mutually commute. 

On the other hand, if the state is taken from a subregion of a string-net ground state evolved by a local unitary of depth $d$, $\Sigma(Z)$ should be isomorphic to the information convex set defined with respect to the string-net ground state (Proposition~\ref{prop:info_convex_u_no_u}). But if the string-net model is non-Abelian, there exists at least one sector $\Sigma^a_{bcd}(Z)$ with $N_{bcd}^a>1$ given by Eq.~\eqref{eq:higher_fusion_multiplicity} and that sector is isomorphic to the full state space on a Hilbert space of dimension greater than one; in particular it contains non-commuting density matrices, inconsistent with Eq.~\eqref{eq:EB_trivial_fusion}. 
\end{proof}

%{\color{red} IK: Here is an attempt to show that EB anyon data coincides with that of string-net data. In fact, for our purpose, it suffices to show that the EB anyon data includes the string-net data. That is, if there is an anyon of quantum dimension $d_a>1$, there is an extreme point of the information convex set with entropy $(S(\rho_a) - S(\sigma))/2 = \log d_a$. Because the $d_a$ and $N_{ab}^c$ defined in EB satisfies the fusion rule axioms, this should be enough. Now the question is, is there a rigorous work that does this calculation on string-net? Two things need to be done: (i) calculation of entanglement entropy (ii) showing that energy over an annulus surrounding the charge is 0. (Zero energy is sufficient because then one can invoke LTQO to argue that over ball-like regions intersecting with the annulus they are indistinguishable with the ground state reduced density matrices.) These are surely ``known,'' but I don't know of a good reference...}

%\SV{Let's clearly state in Sec. 2.1 that we assume from any string-net state ethat the entanglement bootstrap can be used to extract topological data (e.g. fusion multiplicities, $S$-matrix elements) which matches that of the topological order; Proposition 1 (Eq. 9) then establishes that the bootstrap can extract the correct data for any string-net + FDLU.  Then, in this section, we can invoke this assumption, and also state in a footnote that even without this, the fusion multiplicity extracted by the bootstrap for a non-Abelian TQO on a three-hole-punctured disk is larger than one}

In higher spatial dimensions, any state with strict area law (i.e. there exist constants $c$ and $\gamma$ such that $S(A)=c|\partial A|-\gamma$ for all smooth ball-like regions) satisfy the higher-dimensional generalizations of the entanglement bootstrap axioms in Fig.~\ref{fig:EB_axioms_old}. When evolved by a constant-depth local unitary, they satisfy the weaker set of axioms---the higher dimensional generalization of Fig.~\ref{fig:EB_axioms_new}. On such states, one can analogously perform entanglement bootstrap analysis. The relevant structures are more diverse: one encounters a richer variety of topological types of excitations and of the probe regions that detect them~\cite{Huang:2021gsc,Shi:2023kwr}. Nevertheless, within the entanglement bootstrap framework the concept of non-trivial fusion space is defined in the same spirit as in two dimensions: whenever the information covex set contains a sector that is isomorphic to a state space with Hilbert space dimension greater than one, we say that there is a non-trivial fusion space. Again by Lemma~\ref{lemma:commutation_bijection}, a non-trivial fusion space implies the existence of non-commuting density matrices in the information convex set, inconsistent with stabilizerness. By our result, any state with strict area law that supports non-trivial fusion spaces cannot be realized by a stabilizer state, even up to constant-depth local unitaries. A concrete three-dimensional example is provided by the $S_3$ quantum double, where non-trivial fusion spaces arise for regions with nontrivial topology such as a ball minus a trefoil~\cite{Huang:2021gsc}.

\subsection{Quantized mutual braiding in stabilizer projection states}
In the following sections, we demonstrate that the mutual braiding phases in stabilizer projection states satisfying the axioms of Sec.~\ref{subsec:stab_projection} are constrained. We first show in Sec.~\ref{subsec:extreme_points} that given a stabilizer projection state, the extreme points of its information convex set are stabilizer projection states, whose stabilizer generators only differ by phases. Thus, Pauli operators can be used to relate these extreme points, as we show in Sec.~\ref{subsec:EB_anyon_string}. Properly constructing anyon strings using such operators, we show that the mutual braiding phases are quantized by the qudit dimension, demonstrated in Sec.~\ref{subsec:EB_S_matrix}.
\subsubsection{Extreme points of the information convex set}\label{subsec:extreme_points}
%In the following sections, we demonstrate that the $S$-matrix elements of stabilizer projection states satisfying the axioms of Sec.~\ref{subsec:stab_projection} are constrained. 
%We first show that given a stabilizer projection state, the extreme points of its information convex set are stabilizer projection states, whose stabilizer generators only differ by phases. 
We begin by characterizing the extreme points of the information convex set of stabilizer projection states. Going back to the sum over $\ell$ in Eq.~\eqref{eq:EB_expansion_of_density_matrix}, several $\ell$'s are equivalent when multiplied with $\Pi_{S_{\Om^+}(\Om)}$: $\Pi_{S_{\Om^+}(\Om)}\ell=\Pi_{S_{\Om^+}(\Om)}$ if $\ell\in S_{\Om^+}(\Om)$. To avoid redundancy, we can quotient over the subgroup $S_{\Om^+}(\Om)$, analogous to the quotient algebra in~\cite{Haah:2014iuw}. More specifically, let $\calG(S_{\Om^+}(\Om))$ be the generators of $S_{\Om^+}(\Om)$, and $\calG(S_r(\Om))$ be the generators of $S_r(\Om)$ that contains $\calG(S_{\Om^+}(\Om))$. Define $\calL$ as the stabilizer group generated by $\calG(\calL):=\calG(S_r(\Om))\backslash\calG(S_{\Om^+}(\Om))$. Then Eq.~\eqref{eq:EB_expansion_of_density_matrix} can be rewritten as
\begin{equation}\label{eq:EB_general_density_matrix}
\rho_\Om=\Pi_{S_{\Om^+}(\Om)}\sum_{\ell\in\calL}c_{\ell}\ell.
\end{equation}
The reduced density matrix of the reference state on $\Om$ is the stabilizer projection state associated with the stabilizer group $S_r(\Om)$, and therefore can be written as
\begin{equation}
\sigma_\Om\propto\Pi_{S_{\Om^+}(\Om)}\prod_{g\in\calG(\calL)}P(g).
\end{equation}
For each generator $g\in\calG(\calL)$, let $\delta_g$ denote its order (the smallest integer with $g^{\delta_g}\propto I$) and let $\varsigma_g=e^{2\pi i/\delta_g}$. The projector onto the $\varsigma_g^{u}$-eigenspace of $g$ is
\begin{equation}\label{eq:EB_projector}
P(\varsigma_g^{-u}g)=\frac{1}{\delta_g}\sum_{v=1}^{\delta_g}(\varsigma_g^{-u}g)^v.
\end{equation}
which can be viewed as the Fourier transform of the generators. Consider the following family of stabilizer projection states:
\begin{equation}
\rho_\Om(\boldu)\propto\Pi_{S_{\Om^+}(\Om)}\prod_{g\in\calG(\calL)}P(\varsigma_g^{-u(g)}g),
\end{equation}
parametrized by $u(g)$, which encodes a change of phase for the generator $g$. The $u(g)=0$ case is the reference state $\sigma_\Om$. We shall prove that these stabilizer projection states are the extreme points of $\Sigma(\Om)$. 

First, we show that every state of the form~\eqref{eq:EB_general_density_matrix} can be written as a convex mixture of the stabilizer projection states $\rho_\Om(\boldu)$. By inverting~\eqref{eq:EB_projector}, we can express any power of $g$ as a linear combination of the projectors:
\begin{equation}
g^u=\sum_{v=1}^{\delta_g}\varsigma_g^{uv}P(\varsigma_g^{-v}g).
\end{equation}
Slightly generalizing, we can express any element of the stabilizer group as
\begin{equation}
\prod_g g^{u(g)}=\prod_g\sum_{v(g)}\varsigma_g^{u(g)v(g)}P(\varsigma_g^{-v(g)}g).
\end{equation}
It follows that any density matrix in $\Sigma(\Om)$ admits an expansion of the form
\begin{equation}
\rho_\Om=\Pi_{S_{\Om^+}(\Om)}\sum_{\boldu}\alpha_{\boldu} \prod_gP(\varsigma_g^{-u(g)}g).
\end{equation}
The coefficients $\alpha_{\boldu}=\Tr\left[\rho_\Om\prod_gP(\varsigma_g^{-u(g)}g)\right]$ must be real numbers in the range of 0 and 1, which means $\rho_\Om$ is a convex mixture of the states $\rho_\Om(\boldu)$. 

We have shown that any element of $\Sigma(\Omega)$ is a convex mixture of the candidates $\rho_\Omega(\boldu)$. The second step is to show that each candidate $\rho_\Omega(\boldu)$ indeed belongs to $\Sigma(\Omega)$. By definition of $\Sigma(\Om)$, it suffices to construct a density matrix on the thickened region $\Om^+$ that agrees with $\sigma$ on all balls in $\Om^+$ and reduces to $\rho_\Om(\boldu)$ on $\Om$. For this purpose, we aim to show that starting from the reference state on $\Om^+$, the stabilizer subgroup $\calL$ is independent with $S_{\Om^+}$. If that is true, then re-phasing the generators of $\calL$ doesn't affect the stabilizers of $S_{\Om^+}$. The reference state on $\Om^+$ is
\begin{equation}
\sigma_{\Om^+}\propto \prod_{g\in \calG(S_r(\Om^+))}P(g).
\end{equation}
In the stabilizer group $S_r(\Om^+)$, the stabilizers $\calL$ and the stabilizers $S_{\Om^+}(\Om)$ are independent of each other by definition of $\calL$. Further, one can show that $\calL$ and $S_{\Om^+}$ are independent. Indeed, if some nontrivial element of $\calL$ could be trivialized by multiplying with an element $s\in S_{\Om^+}$, then $s$ would have to be supported on $\Om$, because every element of $\calL$ is supported on $\Om$. It would then follow that $s\in S_{\Om^+}(\Om)$, contradicting the independence of $\calL$ and $S_{\Om^+}(\Om)$. Therefore, we may choose an independent generating set of $S_r(\Om^+)$ that contains both $\calG(S_{\Om^+})$ and $\calG(\calL)$. With this choice of generators, the reference state can be written as
\begin{equation}
\sigma_{\Om^+}\propto \prod_{h\in\calG(S_{\Om^+})}P(h)\prod_{g\in\calG(\calL)}P(g)\prod_{k\in\calG(S_r(\Om^+))\backslash[\calG(S_{\Om^+})\cup\calG(\calL)]}P(k).
\end{equation}
Re-phasing the generators $\calG(\calL)$ gives a family of stabilizer projection states
\begin{equation}
\rho_{\Om^+}(\boldu)\propto \prod_{h\in\calG(S_{\Om^+})}P(h)\prod_{g\in\calG(\calL)}P(\varsigma_g^{-u(g)}g)\prod_{k\in\calG(S_r(\Om^+))\backslash[\calG(S_{\Om^+})\cup\calG(\calL)]}P(k).
\end{equation}
The reduced state on $\Om$ is a stabilizer projection state stabilized by $S_{\Om^+}(\Om)$ together with $\varsigma_g^{-u(g)}g$ for every $g\in\calG(\calL)$, which are the stabilizers of $\rho_\Om(\boldu)$. Moreover, $\Tr_{\Om^+\backslash\Om}\rho_{\Om^+}(\boldu)$ has no other stabilizers because the re-phasing process doesn't introduce extra stabilizers. Thus it must agree with $\rho_\Om(\boldu)$:
\begin{equation}
\Tr_{\Om^+\backslash\Om}\rho_{\Om^+}(\boldu)=\rho_\Om(\boldu).
\end{equation}
To conclude that $\rho_\Om(\boldu)\in\Sigma(\Om)$, it remains to verify that $\rho_{\Om^+}(\boldu)$ reproduces the correct reduced density matrices on all balls $b\subseteq\Om^+$. This is true because $\rho_{\Om^+}(\boldu)$ is stabilized by $S_{\Om^+}$, and on each ball $b$ there are no other stabilizers than $S_r(b)$.

We have therefore shown that the extreme points of $\Sigma(\Om)$ are precisely the stabilizer projection states obtained by re-phasing the generators of $\calL$.

\subsubsection{Excitations are Pauli operators}\label{subsec:EB_anyon_string}
We show that there must exist a Pauli operator that maps between two extreme points. This is a consequence of the general fact for stabilizer groups in arbitrary qudit dimensions: the independent set of stabilizer generators can be re-phased by conjugating a suitable Pauli operator. The proof of this statement is provided in appendix~\ref{sec:generators_rephased_by_paulis}.
\begin{lemma}\label{thm:EB_pauli_connects_extreme_points}
For any $\boldu$ and $\boldv$, there exists a Pauli operator $P_{\boldu\rightarrow\boldv}$ supported on $\Om$ such that
\begin{equation}
P_{\boldu\rightarrow\boldv}\rho_\Om(\boldu)P_{\boldu\rightarrow\boldv}^\dagger=\rho_\Om(\boldv).
\end{equation}
\end{lemma}

We now focus on two spatial dimensions and properly define anyon string operators based on~\cite{Shi:2019mlt}. We use the lemma above to show that anyon strings may be chosen to be Pauli operators.

Suppose the reference state on the disk $D$ is equal to the reduced density matrix of the  string-net ground state $\ket{\psi}$ evolved by a depth-$d$ local unitary $U_{\LU_d}$. To simplify the analysis, we take a purifying state $\ket{\tilde{\varphi}}$ of $\sigma_D$ on a sphere $S^2$, such that the state satisfies the entanglement-bootstrap axioms stated in Sec.~\ref{subsec:stab_projection}, namely the decoupling property and local extendibility. It is easy to construct such a purification. Let $D^+$ denote the thickening of $D$, and let $D^{++}$ denote the thickening of $D^+$. Consider a string-net Hamiltonian defined on the sphere, whose local terms agree with those of the original Hamiltonian on an enlarged disk $D^+$ containing the lightcone of $D$. Let $|\varphi\rangle$ be the ground state of this sphere Hamiltonian. Since the string-net Hamiltonians satisfy the local topological order condition~\cite{Michalakis:2012kaj} as proved in~\cite{Jones:2023xew}, the ground-state reduced density matrix on $D^+$ depends only on the Hamiltonian terms supported on $D^{++}$---the thickening of $D^+$. Therefore the sphere ground state $|\varphi\rangle$ and the original string-net ground state $\ket{\psi}$ have the same reduced state on $D^+$. Let $U_{D^+}$ denote the restriction of the local unitary $U_{\LU_d}$ to the gates supported on $D^+$. Define
\begin{equation}
|\tilde\varphi\rangle := U_{D^+}\ket{\varphi}.
\end{equation}
Then $|\tilde\varphi\rangle$ is a purification of $\sigma_D$, because two density matrices over the inner disk $D$ still agrees with each other when acted upon by two finite depth unitaries that differ only by the gate content outside of $D^+$.

Inside the disk $D$, take two holes $A$ and $\bar{A}$. Their complement on the sphere is an annulus $X:=S^2\backslash(A\cup\bar{A})$:
\begin{equation*}
\centering
\includegraphics[width=7.4cm]{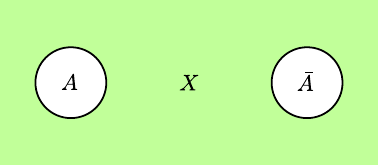}
\end{equation*}
Consider a strip region $V$ connecting the two holes:
\begin{equation*}
\centering
\includegraphics[width=6cm]{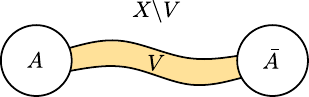}
\end{equation*}
On the annulus $X$, if we trace out the strip $V$ connecting the two holes, the remaining piece $X\backslash V$ has the topology of a disk. As shown in Sec.~\ref{subsec:stab_projection}, the density matrix on a disk-like region is uniquely fixed by the marginals. Therefore, the reduced state over $X\backslash V$ is the same as the reference state $\ket{\tilde{\varphi}}$. By Uhlmann's theorem, there is a unitary $U_{AV\bar{A}}^{(a,\bar{a})}$ supported on $AV\bar{A}$ that maps the reference state to the state with the anyons:
\begin{equation}
U_{AV\bar{A}}^{(a,\bar{a})}\ket{\tilde{\varphi}}=\ket{\tilde{\varphi}^{(a,\bar{a})}}.
\end{equation}
We take such unitaries to be our definition of anyon strings. They are deformable in the sense that they can be defined for any strip $V$ connecting the holes, and their action on the reference state are independent of the precise shape of the strip. Moreover, it doesn't depend on the purification $\ket{\tilde{\varphi}}$ of $\sigma_D$ on the sphere, because any two valid purifications are related by some unitary acting on the complement of $D$, not affecting the resulting $U_{AV\bar{A}}^{(a,\bar{a})}$. The anyon string does depend on the purification $\ket{\tilde{\varphi}^{(a,\bar{a})}}$ we chose for $\sigma^{(a)}_X$, but the $S$ matrix that we construct later is independent of the purification.

Under the assumption that the reference state is a stabilizer state, we show that there is a choice of the anyon strings such that they are Pauli's.

Take a two-holed disk $Y$ inside $AV\bar{A}$, shown in the figure below, colored in blue:
\begin{equation*}
\centering
\includegraphics[width=6cm]{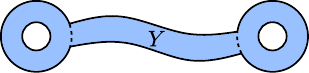}
\end{equation*}
The state with anyons $a$ and $\bar{a}$ inside the two holes is an extreme point $\rho_Y^{a\bar{a}1}$ of $\Sigma(Y)$. From Lemma~\ref{thm:EB_pauli_connects_extreme_points}, there exists a Pauli $P_{(111)\rightarrow(a\bar{a}1)}$ that maps the vacuum state on $Y$ to this extreme point:
\begin{equation}
P_{(111)\rightarrow(a\bar{a}1)}\sigma_YP_{(111)\rightarrow(a\bar{a}1)}^\dagger=\rho_Y^{a\bar{a}1}.
\end{equation}
If we apply $P_{(111)\rightarrow(a\bar{a}1)}$ to the global state $\ket{\tilde{\varphi}}$, it must be a valid purification of the extreme point $\rho_X^{(a)}$ because of the following. Consider the annulus $\tilde{X}$ that is the complement of the two inner holes of $Y$:
\begin{equation*}
\centering
\includegraphics[width=7.4cm]{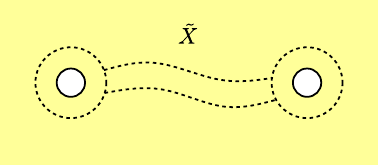}
\end{equation*}
Since $P_{(111)\rightarrow(a\bar{a}1)}$ commutes with all stabilizers of balls inside $Y^+$, it commutes with all stabilizers of balls in $\tilde{X}^+$. Thus, the resulting reduced state over $\tilde{X}$ is inside $\Sigma(\tilde{X})$. If we look at the small annulus $Y_1$ surrounding the inner hole of $Y$,
\begin{equation*}
\centering
\includegraphics[width=6.9cm]{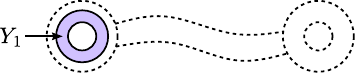}
\end{equation*}
the reduced density matrix will be the extreme point labeled by anyon $a$: $\Tr_{Y\backslash Y_1}\rho_Y^{a\bar{a}1}=\rho_X^{(a)}$. $Y_1$ is also an annulus inside $\tilde{X}$ that surrounds one of the boundaries of $\tilde{X}$. Knowing that the state over $Y_1$ is the extreme point labeled by $a$, we conclude that the reduced state on $\tilde X$ has charge $a$ around the corresponding boundary component. Since annulus sectors are labeled by this charge, the state on $\tilde X$ is the annulus extreme point $\rho_{\tilde X}^{(a)}$. Finally, $X$ and $\tilde X$ are isotopic annuli. By the isomorphism theorem for the information convex set, together with Proposition~\ref{prop:info_convex_u_no_u}, the corresponding information convex sets are isomorphic to each other. Hence the reduced state on $X$ is the extreme point $\rho_X^{(a)}$. By the definition of anyon strings, we conclude that  $P_{(111)\rightarrow(a\bar{a}1)}$ is a valid choice of anyon string.

\subsubsection{Mutual braiding phases are quantized by the qudit dimension}\label{subsec:EB_S_matrix}

Knowing that the excitations can be represented by Pauli strings, one naturally arrives at the conclusion that the mutual braiding phases are quantized: in the braiding experiment of Eq.~\eqref{eq:braiding_experiment}, crossing two Pauli strings can only produce an integer power of $\omega=e^{2\pi i/q}$. The following proposition formalizes this constraint in terms of the $S$-matrix.

\begin{proposition}\label{thm:EB_Smatrix_phases}
Stabilizer projection states cannot realize subsystems of an Abelian string-net ground state, even up to a constant-depth local unitary, unless all the phases of the $S$-matrix elements are integer powers of $\exp(2\pi i/q)$. In other words, the mutual braiding phases are $q$-th roots of unity.
\end{proposition}
To prove the proposition, we first define the $S$-matrix intrinsically from the state, following~\cite{Shi:2019ngw}. We then prove that this definition is independent of the choices of the excitation operators. This allows us to evaluate the same $S$-matrix using Pauli representatives of the anyon strings, for which the braiding phases are necessarily $q$-th root of unity.

To define the $S$-matrix, we introduce another pair of holes $B$, $\bar{B}$. We also define strip regions $V_u$, $V_d$, $W_u$, $W_d$ connecting the holes, see figure below. Analogously, we choose a purification $|\tilde{\varphi}^{(b,\bar{b})}\rangle$ for the $b$ anyons and define the anyon strings accordingly. Now we have defined four anyon strings: $U_{AV_u\bar{A}}^{(a,\bar{a})}$ and $U_{AV_d\bar{A}}^{(a,\bar{a})}$, $U_{BW_u\bar{B}}^{(b,\bar{b})}$ and $U_{BW_d\bar{B}}^{(b,\bar{b})}$, acting on $AV_u\bar{A}$, $AV_d\bar{A}$, $BW_u\bar{B}$ and $BW_d\bar{B}$ respectively. 
\begin{equation*}
\centering
\includegraphics[width=8cm]{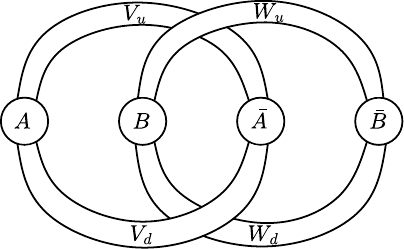}
\end{equation*}
The $S$-matrix element is defined as
\begin{equation}\label{eq:EB_Smatrix}
S_{ab}=\bra{\tilde{\varphi}}U_{BW_u\bar{B}}^{(b,\bar{b})\dagger}U_{AV_u\bar{A}}^{(a,\bar{a})\dagger}U_{AV_d\bar{A}}^{(a,\bar{a})}U_{BW_d\bar{B}}^{(b,\bar{b})}\ket{\tilde{\varphi}}.
\end{equation}
\begin{proposition}\label{thm:EB_Smatrix_uniqueness}
The $S$-matrix defined above is unique in the following sense:
\begin{itemize}
\item It doesn't depend on the location of the strips, as long as $V_u$ and $W_u$ has no overlap with $V_d$ and $W_d$. 
\item It doesn't depend on the purifications $\ket{\tilde{\varphi}^{(a,\bar{a})}}$ and $|\tilde{\varphi}^{(b,\bar{b})}\rangle$.
\end{itemize}
\end{proposition}
\begin{proof}
First, we show that for fixed purifications $\ket{\tilde{\varphi}^{(a,\bar{a})}}$ and $|\tilde{\varphi}^{(b,\bar{b})}\rangle$, $S_{ab}$ is independent of the location of the strips. In fact, one can show that the state
\begin{equation}
U_{AV_d\bar{A}}^{(a,\bar{a})}U_{BW_d\bar{B}}^{(b,\bar{b})}\ket{\tilde{\varphi}}
\end{equation}
is invariant under deformations of the strips. Let's deform the strips $V_d$ and $W_d$ into $V_d'$ and $W_d'$ in such a way that they don't intersect with $V_u$ and $W_u$.
\begin{equation}\begin{aligned}
U_{AV_d'\bar{A}}^{(a,\bar{a})}U_{BW_d'\bar{B}}^{(b,\bar{b})}\ket{\tilde{\varphi}}&=U_{AV_d'\bar{A}}^{(a,\bar{a})}U_{BW_u\bar{B}}^{(b,\bar{b})}\ket{\tilde{\varphi}}=U_{BW_u\bar{B}}^{(b,\bar{b})}U_{AV_d'\bar{A}}^{(a,\bar{a})}\ket{\tilde{\varphi}}=U_{BW_u\bar{B}}^{(b,\bar{b})}U_{AV_d\bar{A}}^{(a,\bar{a})}\ket{\tilde{\varphi}}\\
&=U_{AV_d\bar{A}}^{(a,\bar{a})}U_{BW_u\bar{B}}^{(b,\bar{b})}\ket{\tilde{\varphi}}=U_{AV_d\bar{A}}^{(a,\bar{a})}U_{BW_d\bar{B}}^{(b,\bar{b})}\ket{\tilde{\varphi}}.
\end{aligned}\end{equation}
It follows from $U_{AV_d\bar{A}}^{(a,\bar{a})}\ket{\tilde{\varphi}}=U_{AV_d'\bar{A}}^{(a,\bar{a})}\ket{\tilde{\varphi}}=\ket{\tilde{\varphi}^{(a,\bar{a})}}$, $U_{BW_d'\bar{B}}^{(b,\bar{b})}\ket{\tilde{\varphi}}=U_{BW_u\bar{B}}^{(b,\bar{b})}\ket{\tilde{\varphi}}=|\tilde{\varphi}^{(b,\bar{b})}\rangle$ by definition, and $[U_{AV_d'\bar{A}}^{(a,\bar{a})},U_{BW_u\bar{B}}^{(b,\bar{b})}]=[U_{AV_d\bar{A}}^{(a,\bar{a})},U_{BW_u\bar{B}}^{(b,\bar{b})}]=0$ because of non-overlapping support.

The second step is to show that $S_{ab}$ is independent of the purifications $\ket{\tilde{\varphi}^{(a,\bar{a})}}$ and $|\tilde{\varphi}^{(b,\bar{b})}\rangle$. For different purifications, because of Uhlmann's theorem, the purifying states are multiplied by some unitary acting on $A\bar{A}$ and $B\bar{B}$ respectively:
\begin{equation}
\ket{\tilde{\varphi}^{(a,\bar{a})}}\rightarrow U_{A\bar{A}}\ket{\tilde{\varphi}^{(a,\bar{a})}},\qquad |\tilde{\varphi}^{(b,\bar{b})}\rangle\rightarrow U_{B\bar{B}}|\tilde{\varphi}^{(b,\bar{b})}\rangle.
\end{equation}
The anyon strings becomes
\begin{equation}
U_{AV_u\bar{A}}^{(a,\bar{a})}\rightarrow U_{A\bar{A}}U_{AV_u\bar{A}}^{(a,\bar{a})},\quad
U_{AV_d\bar{A}}^{(a,\bar{a})}\rightarrow U_{A\bar{A}}U_{AV_d\bar{A}}^{(a,\bar{a})},\quad
U_{BW_u\bar{B}}^{(b,\bar{b})}\rightarrow U_{B\bar{B}}U_{BW_u\bar{B}}^{(b,\bar{b})},\quad
U_{BW_d\bar{B}}^{(b,\bar{b})}\rightarrow U_{B\bar{B}}U_{BW_d\bar{B}}^{(b,\bar{b})}.
\end{equation}
Plugging in the definition of the $S$-matrix, the $U_{A\bar{A}}$ cancels with $U_{A\bar{A}}^\dagger$ and $U_{B\bar{B}}$ cancels with $U_{B\bar{B}}^\dagger$.
\end{proof}

Because of the uniqueness of the $S$-matrix, we can choose the anyon strings to be the native anyon strings of the underlying topological order. Hence the $S$-matrix we defined here must agree with the $S$-matrix of the topological order.
\begin{proposition}
If the reference state is a subsystem of a string-net ground state evolved by finite-depth local unitary, then the $S$-matrix defined by Eq.~\eqref{eq:EB_Smatrix} must agree with the native $S$-matrix of the underlying topological order.
\end{proposition}
\begin{proof}
We show that the native anyon strings in the anyon theory are valid candidates of anyon strings. Take any anyon string operator of the string-net model (conjugated by the finite-depth circuit of course) that creates $a$ somewhere inside $A$ and $\bar{a}$ somewhere inside $\bar{A}$.
\begin{equation*}
\centering
\includegraphics[width=6cm]{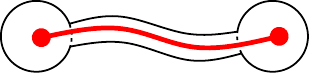}
\end{equation*}
As we discussed in Sec.~\ref{subsec:nonabelian_eb}, Refs.~\cite{bols2025sector,bols2026sector} demonstrate that the topological data extracted from entanglement bootstrap agrees with the intrinsic topological data of the underlying string-net model. In particular, applying the native anyon string creates the extreme point $\sigma_X^{(a)}$ on the annulus $X=S^2\backslash(A\cup\bar{A})$. The location of the endpoints of the string is not important---this merely changes the purification $\ket{\tilde{\varphi}^{(a,\bar{a})}}$, but does not affect the $S$-matrix (Proposition~\ref{thm:EB_Smatrix_uniqueness}). Moreover, since the complement of $AV\bar{A}$ is in the vacuum sector, the action of the string on the ground state can be implemented by a unitary supported on $AV\bar{A}$.

Under this choice, the $S$-matrix defined by~\eqref{eq:EB_Smatrix} recovers the $S$-matrix of the underlying topological order. Since the $S$-matrix is uniquely defined, this completes the proof.
\end{proof}

We can now prove Proposition~\ref{thm:EB_Smatrix_phases}. 
\begin{proof}[Proof of Proposition~\ref{thm:EB_Smatrix_phases}]
By the previous subsection and the purification-independence just proved, we are free to choose convenient representatives of the anyon strings. In particular, we may choose all four strings to be Pauli operators:
\begin{equation}
P_{AV_u\bar A}^{(a,\bar a)},\quad P_{AV_d\bar A}^{(a,\bar a)},\quad P_{BW_u\bar B}^{(b,\bar b)},\quad P_{BW_d\bar B}^{(b,\bar b)}.
\end{equation}
Since Pauli operators commute up to integer powers of $\omega=e^{2\pi i/q}$, there exists an integer $r$ such that
\begin{equation}
P_{BW_u\bar B}^{(b,\bar b)\dagger} P_{AV_u\bar A}^{(a,\bar a)\dagger} P_{AV_d\bar A}^{(a,\bar a)} P_{BW_d\bar B}^{(b,\bar b)}
= \omega^r\, P_{AV_u\bar A}^{(a,\bar a)\dagger} P_{AV_d\bar A}^{(a,\bar a)} P_{BW_u\bar B}^{(b,\bar b)\dagger} P_{BW_d\bar B}^{(b,\bar b)}.
\end{equation}
In other words, the $S$-matrix can be unlinked by a phase. Acting on $\ket{\tilde{\varphi}}$, the products are the identity because the upper and lower strings of the same anyon type act identically. Therefore
\begin{equation}
S_{ab} = \bra{\tilde{\varphi}} P_{BW_u\bar B}^{(b,\bar b)\dagger} P_{AV_u\bar A}^{(a,\bar a)\dagger} P_{AV_d\bar A}^{(a,\bar a)} P_{BW_d\bar B}^{(b,\bar b)} \ket{\tilde{\varphi}} = \omega^r.
\end{equation}
Thus every $S$-matrix phase is an integer power of $\exp(2\pi i/q)$.
\end{proof}

\section{Extensive long-range magic}
In this section, we show that the results of Sec.~\ref{sec:EB} lead to lower bound on long-range magic in various topological states. 
In particular, we prove that non-abelian topologically ordered states (at fixed points, e.g. string-net states) have extensive long-range magic.

\subsection{Magic measures}\label{subsec:LRM_magic_measures}

Here, we review the magic measures that we consider in this work.
These measures are magic monotones, which are well-motivated from the resource theory as proper measures of magic~\cite{veitch2014resource}.
In the proofs, we mainly focus on the log-stabilizer fidelity and show that this quantity is extensive in various setups.
This implies that other magic measures also exhibit extensive scalings, as the log-stabilizer fidelity is the lower bound on other magic measures as in Proposition~\ref{thm:LRM_magic_measures_order}.
We use $\calS$ to denote the set of states with no magic, i.e. mixtures (convex sums) of pure stabilizer states. 

% Here we list the magic measures that we will use. These magic measures are magic monotones that are well motivated in resource theory. 

\begin{itemize}
\item Log-stabilizer fidelity~\cite{Bravyi:2018ugg}. 

The stabilizer fidelity of state $\rho$ is defined as its maximum fidelity with a stabilizer mixture $\sigma \in \mathcal{S}$,
\begin{equation}
\begin{aligned}
&\calF(\rho,\calS):=\max_{\sigma\in\calS}\calF(\rho,\sigma),\qquad \calF(\rho,\sigma)=\lVert\sqrt{\rho}\sqrt{\sigma}\rVert_1^2= F(\rho,\sigma)^2.
\end{aligned}
\end{equation}
We consider the extensive quantity given by the logarithm of stabilizer fidelity
\begin{equation}
\LF(\rho):=-\log\calF(\rho,\calS).
\end{equation}

\item Relative entropy of magic~\cite{Emerson:2013zse}.

An alternative measure of magic is the relative entropy of magic, given by the minimum relative entropy between the state $\rho$ and a mixture $\sigma$ of stabilizer states,
\begin{equation}
S(\rho\|\calS):=\min_{\sigma\in\calS}S(\rho\|\sigma),\qquad S(\rho\|\sigma)=\Tr(\rho\log\rho)-\Tr(\rho\log\sigma).
\end{equation}

One can also quantify magic using R\'enyi relative entropies. In particular, we consider the max-relative entropy of magic~\cite{Wang:2020nyx},
\begin{equation}
S_{\max}(\rho\|\calS):=\min_{\sigma\in\calS}S_{\max}(\rho\|\sigma),\qquad S_{\max}(\rho\|\sigma)=\log \min\{\lambda\ge 0:\rho\preceq2^\lambda\sigma\}.
\end{equation}

\item Log-robustness of magic, and log-generalized robustness of magic~\cite{Howard:2017maw,Seddon:2021rij,Liu:2020yso}.

Another set of magic measures is the robustness of magic. 
We consider the mixture of the state $\rho$ and another state $\sigma \in \calS$ (a mixture of stabilizer states), i.e. $\sigma' = (\rho + s\sigma)/(1+s)$.
The robustness of magic is defined as the smallest $s$ such that the mixture $\sigma'$ has no magic.

Formally, the robustness of magic $\R(\rho)$ and its logarithmic version $\Lr(\rho)$ are defined as
\begin{equation}\begin{aligned}
\Lr(\rho)&=\log[1+2\R(\rho)],\; \operatorname{R}(\rho):=\min\{s\ge0:\exists\sigma,\sigma'\in\calS\ \st\ \rho=(1+s)\sigma'-s\sigma\}, 
\end{aligned}\end{equation}
There also exists a generalized version of the robustness $\GR(\rho)$, where $\sigma \in \calD$ is allowed to be chosen from all density matrices $\calD$.
Specifically,
\begin{align}
\Lgr(\rho)&=\log[1+2\GR(\rho)],\; \GR(\rho):=\min\{s\ge0:\exists\sigma'\in\calS,\sigma\in \calD \ \st\ \rho=(1+s)\sigma'-s\sigma\}.
\end{align}
\end{itemize}

\begin{proposition}\label{thm:LRM_magic_measures_order}
\begin{equation}
\LF(\rho)\le S(\rho\|\calS)\le S_{\max}(\rho\|\calS)\le\Lgr(\rho)\le\Lr(\rho)
\end{equation}
The magic measures above are at most extensive. Consider a system of $n$ qudits of local dimension $q$. For an arbitrary state, these magic measures are upper bounded by $(n+2^{-n-1})\log q$.

\end{proposition}
\begin{proof}
The first two inequalities follow from the monotonicity of R\'enyi relative entropy with respect to the R\'enyi index~\cite{Muller-Lennert:2013liu,Beigi:2013teh}. The $\alpha$-R\'enyi relative entropy is defined as~\cite{Wilde:2013bdg,Muller-Lennert:2013liu}:
\begin{equation}
S_\alpha(\rho\|\sigma):=\begin{cases}
\frac{1}{\alpha-1}\log\Tr\left(\sigma^{\frac{1-\alpha}{2\alpha}}\rho\sigma^{\frac{1-\alpha}{2\alpha}}\right)^\alpha,& \text{if }\rho\preceq\sigma \\
\infty,& \text{otherwise}
\end{cases}
\end{equation}
for $\alpha\in[\frac{1}{2},1)\cup(1,\infty)$. The $\alpha=\frac{1}{2}$ case corresponds to fidelity: $D_{\frac{1}{2}}(\rho,\sigma)=-\log\lVert\sqrt{\rho}\sqrt{\sigma}\lVert_1^2=-\log \calF(\rho,\sigma)$. It is monotonically increasing in $\alpha$~\cite{Muller-Lennert:2013liu,Beigi:2013teh}. The $\alpha\rightarrow\infty$ limit is the max-relative entropy:
\begin{equation}
\lim_{\alpha\rightarrow\infty}S_\alpha(\rho\|\sigma)=\log\left\lVert\sigma^{-\frac{1}{2}}\rho\sigma^{-\frac{1}{2}}\right\rVert_\infty
=\log\min\{\lambda\ge 0: \rho\preceq\lambda\sigma\}=S_{\max}(\rho\|\sigma)
\end{equation}
The second equality is because $\sigma^{-\frac{1}{2}}\rho\sigma^{-\frac{1}{2}}\preceq\lambda I\Leftrightarrow \rho\preceq\lambda\sigma$.

For the third inequality, if $s=\GR(\rho)$ and $\rho=(1+s)\sigma_1-s\sigma_2$ with $\sigma_1\in\calS$ and $\sigma_2\in\calD$, then $\rho\preceq (1+2s)\sigma_1$. Hence
\begin{equation}
S_{\max}(\rho\|\calS)\le \log(1+2s)=\Lgr(\rho).
\end{equation}
Finally, the definition of generalized robustness optimizes over a larger set of auxiliary states than the definition of robustness, so $\GR(\rho)\le \R(\rho)$, and therefore $\Lgr(\rho)\le \Lr(\rho)$.
The log-robustness cannot scale faster than extensive and has an upper bound in an $n$-qudit system, given by Proposition~\ref{thm:Rom_upper_bound} in Appendix~\ref{sec:Rom_upper_bound}, generalizing the qubit case proved in~\cite{Liu:2020yso}.
\end{proof}
%\YZ{are there lower bounds to distillable magic?}

\begin{comment}
\YZ{As an example, Haar random states have $O(n)$ magic: for a fixed stabilizer state $\phi$,}
\begin{equation}
\Pr[|\braket{\phi}{\psi}|^2\ge q^{-(1-\alpha)n}]=\Pr[|\braket{\phi}{\psi}|^{2k}\ge q^{-(1-\alpha)nk}]\le q^{(1-\alpha)nk}\Ens_{\ket{\psi}\sim\H}|\braket{\phi}{\psi}|^{2k}\le k!q^{-\alpha nk}
\end{equation}
for $k\ll q^n$. The total number of stabilizer states is... so 
\begin{equation}
\Pr[\LF(\ketbra{\psi})\le (1-\alpha)n\log q]=\Pr[\exists \ket{\psi}\in\calS, \st |\braket{\phi}{\psi}|^2\ge q^{-(1-\alpha)n}]\le
\end{equation}
\end{comment}

\subsection{Long-range magic in non-Abelian topologically ordered states}\label{subsec:LRM_nA}
In this section, we show that the ground states of non-Abelian string-net models have extensive long-range magic. This conclusion holds for any qudit dimension.

\begin{theorem}\label{thm:LRM_non-Abelian}
The ground state $\ket{\psi_{\nA}}$ of a non-Abelian string-net model has long-range magic
\begin{equation}
\LR M_d(\ketbra{\psi_{\nA}})\ge n\cdot \calC_{\nA}(\UFC,q,d),
\end{equation}
in a system of $n$ qudits with qudit dimension $q$, where $\calC_{\nA}(\UFC,q,d)$ only depends on the unitary fusion category, $q$ and circuit depth $d$. The magic measure $M$ is any magic measure listed in Sec.~\ref{subsec:LRM_magic_measures}.
\end{theorem}
\begin{proof}
Consider the state $|\tilde{\psi}_{\nA}\rangle$, obtained by evolving the non-Abelian string-net ground state by a local unitary of depth at most $d$. We choose a family of disjoint ball-like patches $A_1,\dots,A_m$ as in Fig.~\ref{fig:patches_with_magic_ground_state}. 
\begin{figure}
    \centering
    \includegraphics[width=0.6\linewidth]{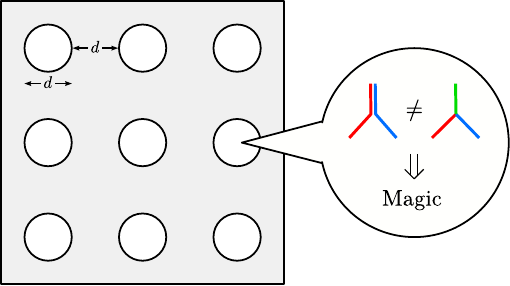}
    \caption{Patch construction for the ground state lower bound. 
After evolving a non-Abelian string-net ground state by a depth-$d$ local unitary, one can choose $O(n/d^2)$ disjoint patches of size $O(d)$ separated by distance $O(d)$. The reduced state on the union of these patches factorizes. Zooming in on each patch, the entanglement bootstrap analysis detects a nontrivial fusion space, which is incompatible with stabilizerness.}
    \label{fig:patches_with_magic_ground_state}
\end{figure}
The patches are separated by a distance of order $d$, with a sufficiently large constant prefactor, so that their backward lightcones remain disjoint. We also choose each patch to have diameter $O(d)$, such that Theorem~\ref{thm:EB_stabilizer_implies_Abelian} in Sec. \ref{subsec:stab_projection} is applicable.
In two spatial dimensions, one can find $m=O(n/d^2)$ such patches.

Because of the $O(d)$ distance between the patches, the reduced state on these patches factorizes: $\rho_{A_1\cup\cdots\cup A_m}=\otimes_{i=1}^m\rho_i$, where $\rho_i$ is the density matrix of the $i$'th patch $A_i$. 
To see why it factorizes, consider the thickening of $A_i$ by the lightcone radius of the depth-$d$ circuit. Before applying the circuit, the mutual information between disjoint disk-like regions in the string-net ground state $|\psi_{\nA}\rangle$ vanishes. Thus, the density matrix over the thickened patches factorizes as a tensor product. Now consider the reduced density matrix of the evolved state on $A_1\cup\cdots\cup A_m$. When evaluating this reduced density matrix, every gate outside the backward lightcone of $A_1\cup\cdots\cup A_m$ cancels against its adjoint after tracing out the complement region. Since the backward lightcones of the patches are disjoint, the circuit acts independently on each thickened patch. It follows that the patches factorize.

In the next step, we make use of the tensor product structure on the patches to show that the root fidelity 
\begin{equation}
F(\ketbra*{\tilde{\psi}_{\nA}},\calS):=\max_{\sigma\in\calS}F(\ketbra*{\tilde{\psi}_{\nA}},\sigma) 
\end{equation}
is exponentially small in the number of patches $m$. Since $|\tilde{\psi}_{\nA}\rangle$ is a pure state, the fidelity $\calF(\ketbra*{\tilde{\psi}_{\nA}},\sigma)=\bra*{\tilde{\psi}_{\nA}}\sigma\ket*{\tilde{\psi}_{\nA}}$ is linear in $\sigma$. Hence the maximum over the convex set $\calS$ is attained at an extreme point, namely at some pure stabilizer state $\ket{\phi}$.\footnote{Equivalently for $F(\ketbra*{\tilde{\psi}_{\nA}},\sigma)=\sqrt{\calF(\ketbra*{\tilde{\psi}_{\nA}},\sigma)}$.} Tracing out the complement of the patches, $\ketbra*{\tilde{\psi}_{\nA}}$ reduces to the product state $\otimes_{i=1}^m\rho_i$, while $\ketbra{\phi}$ reduces to a stabilizer projection state $\sigma\in\SP$. By monotonicity of fidelity under partial trace,
\begin{equation}
F(\ketbra*{\tilde{\psi}_{\nA}},\ketbra{\phi})\le F\!\left(\otimes_{i=1}^m\rho_i,\sigma\right).
\end{equation}
On each patch, we can now apply the entanglement bootstrap analysis. Since the underlying anyon theory is non-Abelian, Theorem~\ref{thm:EB_stabilizer_implies_Abelian} implies that $\rho_{i}$ cannot be a stabilizer projection state.  
% Since the anyons are non-Abelian, we conclude that $\rho_{i}$ cannot be a stabilizer projection state according to Theorem~\ref{thm:EB_stabilizer_implies_Abelian}. 
The minimum trace distance between $\rho_{A_i}$ and stabilizer projection states is lower bounded by a constant that depends on the unitary fusion category, qudit dimension $q$ and the patch size depending on $d$. It is not obvious at this point that the fidelity is exponentially small in $m$, because $\sigma$ can be entangled across the patches and the fidelity isn't simply multiplicative. Nonetheless, in Appendix~\ref{sec:multiplicative_fidelity_bound} we prove Proposition~\ref{thm:multiplicative_fidelity_bound} and~\ref{thm:exp_small_fidelity_bound} by making use of the notion of stabilizer measurement fidelity, showing that the fidelity between the product state $\otimes_{i=1}^m\rho_i$ and any stabilizer projection state is exponentially small in $m$. Therefore $F(\ketbra*{\tilde{\psi}_{\nA}},\calS)$ is exponentially small in the number of patches, and hence exponentially small in system size for fixed $d$.
Finally, Proposition~\ref{thm:LRM_magic_measures_order} implies that every magic measure listed in Sec.~\ref{subsec:LRM_magic_measures} is lower bounded by the same quantity. Since the depth-$d$ local unitary was arbitrary, taking the minimum over all depth-$d$ local unitaries proves the claimed lower bound on $\LR M_d$.
\end{proof}

\subsection{Long-range magic in Abelian topological states depends on qudit dimension}\label{subsec:LRM_abelian}
We now consider Abelian string-net models. We show that if the qudit dimension cannot resolve all braiding phases, even Abelian string-net ground states must have extensive long-range magic. The proof for extensiveness is parallel to Sec.~\ref{subsec:LRM_nA}, but here the local obstruction comes from the quantization of mutual braiding phases established in Sec.~\ref{subsec:EB_S_matrix}. This again highlights the significance of the unconditional (qudit-dimension-independent) existence of extensive long-range magic in the non-Abelian case.
\begin{theorem}
Stabilizer states can only realize Abelian string-net models where mutual braiding phases are $q$-th root of unity. Moreover, if there exist mutual braiding phases that are not $q$-th root of unity, then the Abelian string-net ground state has long-range magic
\begin{equation}
\LR M_d(\ketbra{\psi_{\A}})\ge n\cdot \calC_{\A}(\UFC,q,d),
\end{equation}
where $\calC_{\A}(\UFC,q,d)$ only depends on the unitary fusion category, $q$ and circuit depth $d$. The magic measure $M$ is any magic measure listed in Sec.~\ref{subsec:LRM_magic_measures}.
\end{theorem}

\begin{proof}
Take any disk-like subregion with sufficiently large $O(d)$ radius of the state $\ketbra{\psi_{\A}}$ evolved by a depth-$d$ local unitary. By Proposition~\ref{thm:EB_Smatrix_phases}, if the reduced state over the subregion is a stabilizer projection state, then the mutual braiding phases of the topological order can only be $q$-th root of unity. If not, then the subregion cannot be a stabilizer projection state. The same procedure as in Theorem~\ref{thm:LRM_non-Abelian} gives the claimed lower bound on long-range magic.
\end{proof}

\subsection{Extensive long-range magic in low-energy states}\label{subsec:low_energy_states}

In this section, we show that in non-Abelian string-net models, the extensiveness of long-range magic persists at low energies, as long as the circuit depth $d$ does not scale badly with the energy density. Our statements hold true for any local qudit dimension.
\begin{theorem}
For a non-Abelian string-net Hamiltonian $H$ on $n$ qudits, where each term in the Hamiltonian has $O(1)$ operator norm, any state $\ket{\psi}$ with energy $\bra{\psi}H\ket{\psi}$ less than $\epsilon n$ above the ground state has long-range magic
\begin{equation}
\LR\Lr_d(\ketbra{\psi})=\Omega(nd^{-2}-n\epsilon^{1-\alpha}\log(1/\epsilon))\cdot\calC_{\nA}(\UFC,q,d),
\end{equation}
for any fixed constant $\alpha\in(0,1)$ and sufficiently small $\epsilon$.
\end{theorem}
The idea is that for any pure state with energy bounded by $\epsilon n$ above the ground state, the state has support mostly on low-energy eigenstates. In the anyon theory, low-energy eigenstates contain low density of anyon excitations, which means that there is an extensive number of clean (anyon-free) patches. This persists upon the application of local unitaries, as long as the circuit depth is sufficiently short compared to the typical separation of the anyons. %the circuit depth doesn't scale too badly with the density of anyons.

\begin{proof}
Let us shift the spectrum to set the ground state energy to $0$. Consider a pure state $\ket{\psi}$ with energy bounded by $\epsilon n$, and expand the state in the energy basis: $\ket{\psi}=\sum_ic_i\ket{E_i}$ where $E_i$ label the energies. The amplitudes $|c_i|^2$ define a probability distribution. The probability support on high-energy states is bounded by Markov's inequality: for any fixed $\alpha\in(0,1)$, 
\begin{equation}
\Pr[E\ge\epsilon^{1-\alpha}n]\le\frac{\bar{E}}{\epsilon^{1-\alpha} n}\le\epsilon^\alpha,
\end{equation}
where $\bar{E}=\bra{\psi}H\ket{\psi}\le\epsilon n$. We may therefore decompose
\begin{equation}\label{eq:low_energy_state_decomposition}
\ket{\psi}=c_l\ket{\psi_l}+c_h\ket{\psi_h},\qquad\text{where}\qquad \ket{\psi_l}=\sum_{E_i\le\epsilon^{1-\alpha}n}c_i'\ket{E_i},\qquad|c_h|^2\le\epsilon^\alpha.
\end{equation}

In the anyon theory, energy eigenstates can be labeled by the anyon charges and the fusion tree. Since each anyon excitation has $O(1)$ energy, an eigenstate with energy bounded by $\epsilon^{1-\alpha}n$ has at most $O(\epsilon^{1-\alpha}n)$ anyon endpoints. Let $|\tilde{E}_i\rangle$ be the energy eigenstate evolved by the depth-$d$ local unitary. After applying the local unitary, each endpoint expands to a region of diameter $O(d)$. Therefore one can still find $\Omega(nd^{-2}-n\epsilon^{1-\alpha})$ clean disks of radius $O(d)$ with no anyons and remain separated by distance $O(d)$. The situation is illustrated in Fig.~\ref{fig:LRM_anyons_and_patches}. On each such disk, the reduced density matrix agrees with the vacuum reduced state, so the argument of Sec.~\ref{subsec:LRM_nA} applies exactly as before. 
\begin{figure}
    \centering
    \includegraphics[width=0.6\linewidth]{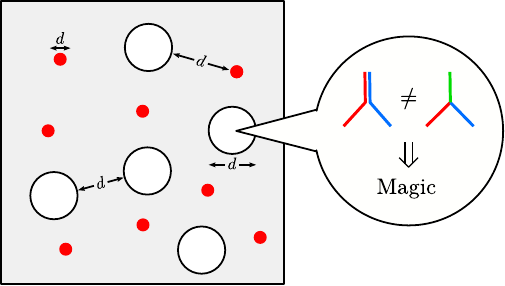}
    \caption{Patch construction for a low-energy eigenstate. An eigenstate with energy at most $\epsilon^{1-\alpha}n$ above the ground state contains $O(\epsilon^{1-\alpha}n)$ anyon endpoints. 
After a depth-$d$ local unitary, each endpoint expands to a region of diameter $O(d)$, shown in red. 
Away from these enlarged excitation regions, one can still choose $\Om(nd^{-2}-n\epsilon^{1-\alpha})$ clean patches of size $O(d)$, mutually separated by distance $O(d)$. 
On each clean patch, the reduced density matrix agrees with the ground state, so the same non-Abelian fusion-space obstruction applies.}
    \label{fig:LRM_anyons_and_patches}
\end{figure}
As a result, the eigenstate has long-range magic: 
\begin{equation}
\LR M_d(\ketbra{E_i})=\Omega(nd^{-2}-n\epsilon^{1-\alpha})\cdot\calC_{\nA}(\UFC,q,d).
\end{equation}
and in particular, the overlap between the evolved state $|\tilde{E}_i\rangle$ and any stabilizer state $\ket{\phi}$ is exponentially small:
\begin{equation}
|\langle\phi|\tilde{E}_i\rangle|^2\le 2^{-\Omega(nd^{-2}-n\epsilon^{1-\alpha})\cdot\calC_{\nA}(\UFC,q,d)}.
\end{equation}
The above bounds are meaningful for $d\le O(\epsilon^{-\frac{1-\alpha}{2}})$.

Next, we use this fact for low-energy eigenstates to lower bound the long-range magic of the low energy piece $\ket{\psi_l}$ in Eq.~\eqref{eq:low_energy_state_decomposition}. Let $|\tilde{\psi}\rangle$, $|\tilde{\psi_l}\rangle$ $|\tilde{E_i}\rangle$ and $|\tilde{\psi_h}\rangle$ be the states in Eq~\eqref{eq:low_energy_state_decomposition} acted upon by the same depth-$d$ local unitary. The number of eigenstates with energy at most $\epsilon^{1-\alpha}n$ above the ground state is at most $2^{O(n\epsilon^{1-\alpha}\log(1/\epsilon))}$. So by Cauchy--Schwarz, 
\begin{equation}\begin{aligned}
|\langle\phi|\tilde{\psi_l}\rangle|&\le\sqrt{\sum_{E_i\le\epsilon^{1-\alpha}n}|c_i'|^2}\sqrt{\sum_{E_i\le\epsilon^{1-\alpha}n}|\braket{\phi}{E_i}|^2} \\
&\le 1\times\sqrt{2^{O(n\epsilon^{1-\alpha}\log(1/\epsilon))}\times 2^{-\Om(nd^{-2}-n\epsilon^{1-\alpha})\cdot\calC_{\nA}(\UFC,q,d)}}\\
&=2^{-\Om(nd^{-2}-n\epsilon^{1-\alpha}\log(1/\epsilon))\cdot\calC_{\nA}(\UFC,q,d)}.
\end{aligned}\end{equation}
Thus $\LF(|\tilde{\psi}_l\rangle\langle\tilde{\psi}_l|,\calS)=\Om(nd^{-2}-n\epsilon^{1-\alpha}\log(1/\epsilon))\cdot\calC_{\nA}(\UFC,q,d)$.

To conclude extensive magic of the full state $|\tilde{\psi}\rangle$, we note that it is a perturbation on the low-energy piece---the strength of the perturbation scales as $\epsilon^\alpha$. We need to be careful with the choice of magic measures here, because it is not generally true that extensiveness of the magic measure is robust under perturbations. Take the log-stabilizer fidelity as an example, if we perturb the state $|\tilde{\psi_l}\rangle$ with a stabilizer state, then the stabilizer fidelity is now controlled by the perturbation strength, which is not exponentially small. To this end, we consider the log-robustness of magic, which can be formulated as a linear program:
\begin{equation}
\Lr(\rho)=\log[1+2\R(\rho)],\qquad 1+2\R(\rho)=\min\bigg\{\sum_\sigma|c_\sigma|\ \st\ \rho=\sum_{\sigma\in\calS}c_\sigma\sigma,\ c_\sigma\in\mathbb{R}\bigg\}.
\end{equation}
The dual program is
\begin{equation}
1+2\R(\rho)=\max_{A}\bigg\{|\Tr(A\rho)|\ \st\ |\Tr(A\sigma)|\le 1,\ \forall\sigma\in\calS\bigg\}.
\end{equation}
In evaluating $\R(|\tilde{\psi}\rangle\langle\tilde{\psi}|)$, let us choose the operator $A$ to be
\begin{equation}
A=\frac{|\tilde{\psi}_l\rangle\langle\tilde{\psi}_l|}{\calF(|\tilde{\psi}_l\rangle\langle\tilde{\psi}_l|,\calS)}.
\end{equation}
It is a valid choice because $|\Tr(A\sigma)|\le 1$, $\forall\sigma\in S$ by definition of the stabilizer fidelity. Plugging in the dual program, we get
\begin{equation}
\begin{aligned}
1+2\R(|\tilde{\psi}\rangle\langle\tilde{\psi}|)&
\ge \Tr(A|\tilde{\psi}\rangle\langle\tilde{\psi}|)=\frac{|\langle\tilde{\psi}_l|\tilde{\psi}\rangle|^2}{\calF(|\tilde{\psi}_l\rangle\langle\tilde{\psi}_l|,\calS)}\\
&\ge (1-\epsilon^\alpha)2^{\Om(nd^{-2}-n\epsilon^{1-\alpha}\log(1/\epsilon))\cdot\calC_{\nA}(\UFC,q,d)},
\end{aligned}
\end{equation}
\begin{equation}
\Lr(|\tilde{\psi}\rangle\langle\tilde{\psi}|)=\Omega(nd^{-2}-n\epsilon^{1-\alpha}\log(1/\epsilon))\cdot\calC_{\nA}(\UFC,q,d)
\end{equation}
for sufficiently small $\epsilon$. Since the depth-$d$ local unitary is arbitrary, this completes the proof.
\end{proof}

\subsection{Upper bounds on magic from entanglement renormalization}\label{subsec:LRM_MERA}
The previous subsections established lower bounds. We now give an intuitive renormalization-based picture, which also yields a constructive upper bound on long-range magic.

Long-range properties are naturally probed by renormalization: one coarse-grains the system to obtain an effective description at a large scale by successively applying unitary gates to disentangle local degrees of freedom. From this viewpoint, long-range magic is the magic that survives repeated coarse-graining. If coarse-graining eventually removes the magic, then it is short-ranged; if the magic persists to the infrared (IR), then it is long-ranged.

In qudit systems, such coarse-graining can be implemented by a MERA circuit~\cite{Aguado:2007oza}. The MERA maps ultraviolet (UV) degrees of freedom to IR degrees of freedom by applying local disentangling that reduces the number of active qudits at each scale. Conversely, it reconstructs the many-body state from an IR state $\ket{\psi_{\IR}}$ by a sequence of isometries. A state is a fixed point of entanglement renormalization if the same local isometries appear at every scale--this is the case for string-net ground states\cite{Koenig:2009koo}.

For topological codes, these isometries define the encoding circuit, while $\ket{\psi_{\IR}}$ is the logical state. The isometries grow a small code into a large one without changing the number of logical qudits. In the $\mathbb{Z}_q$ toric code, for example, the MERA takes the form~\cite{Aguado:2007oza} 
\begin{equation*}
\centering\includegraphics[width=17cm]{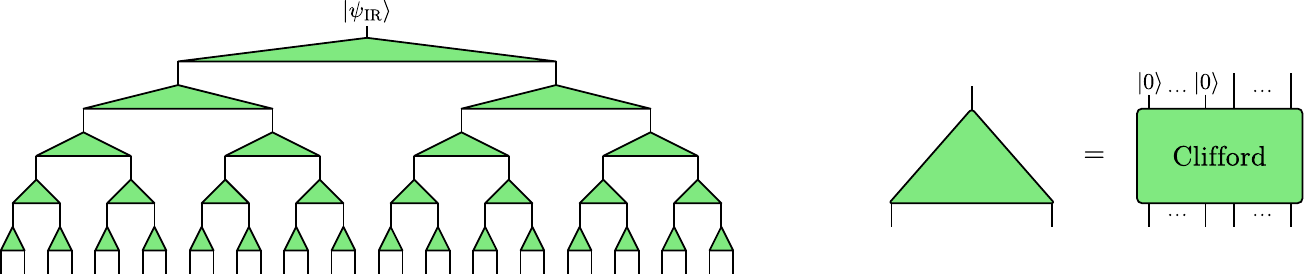}.
\end{equation*}
Each isometry can be implemented by adjoining $\ket{0}$ ancillas and applying a Clifford unitary. Hence if $\ket{\psi_{\IR}}$ is a stabilizer state, the full state has no magic, while if $\ket{\psi_{\IR}}$ is a logical $T$ state, the magic resides in the deep IR and is therefore long-ranged, as shown in~\cite{Wei:2025irp}.

For non-Abelian anyons, the situation is qualitatively different. The MERA unitaries for non-Abelian string-net models~\cite{Koenig:2009koo} are non-Clifford. For example, in quantum doubles, the unitaries involves controlled group multiplications~\cite{Aguado:2007oza}, which are non-Clifford when the group is non-Abelian. Thus non-Abelian string-net states carry magic at every scale:
\begin{equation*}
\centering\includegraphics[width=17cm]{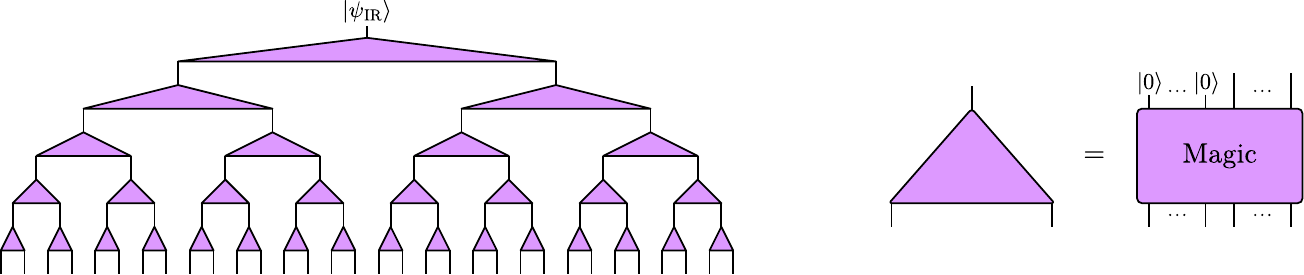}.
\end{equation*}
In particular, unlike the Abelian case, long-range magic in the non-Abelian case does not rely on the choice of state within the space of degenerate ground states.

We emphasize that magic in the IR state does not by itself imply long-range magic in the UV. For example, consider the GHZ state represented as $V_{\MERA}\ket +$, where the MERA isometries are implemented by adjoining $\ket 0$ ancillas and applying CNOTs. If we apply a $T$-gate to one of the qubits, the resulting state $\frac{1}{\sqrt{2}}(\ket{0\cdots 0}+e^{i\pi/4}\ket{1\cdots 1})$ clearly has no long-range magic. But this state can also be realized by replacing the IR state with the $T$ state\footnote{One way to see this is to observe that the MERA is the encoding isometry for a repetition code. The $T$-gate applied to one qubit in the IR state exactly produces the logical $T$-state for this code.}. 

The MERA also gives a natural upper bound on long-range magic, since it provides an explicit coarse-graining procedure that reduces the number of degrees of freedom at each step. For a fixed-point state on a $D$-dimensional lattice, suppose each renormalization step rescales lengths by a factor $a>1$. After $t$ steps, the disentangling gates act on UV degrees of freedom separated by distance $O(a^t)$, so implementing such a step by geometrically local gates requires depth $O(a^t)$. Thus a depth-$d$ local circuit corresponds to coarse-graining down to scale $a^t\sim d$. After $t$ steps, the number of remaining degrees of freedom is $O(n/a^{Dt})=O(n/d^D)$. 
\begin{equation*} 
\centering \includegraphics[width=14cm]{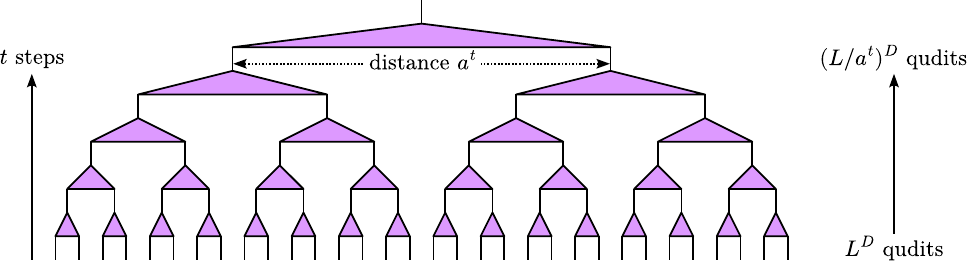} 
\end{equation*} 
Since all the magic measures we consider are at most extensive with respect to the number of degrees of freedom (Proposition~\ref{thm:LRM_magic_measures_order}), this gives an upper bound to the long-range magic at depth $d$. For a fixed-point state $\ket{\psi}$ in $D$ spatial dimensions: 
\begin{equation}
\LR M_d(\ketbra{\psi})\le O(n/d^D),
\end{equation}
where the implicit constant doesn't depend on $n$ and $d$.

\section{Discussion}

Our results suggest that long-range magic provides a useful diagnostic of intrinsic non-stabilizerness in topologically ordered phases, and in particular in non-Abelian topological quantum order. Several directions naturally follow from this work.

First, it is useful to view \emph{no low-energy trivial magic} (NLTM) as a strengthening of the NLTS paradigm, in which every sufficiently low energy-density state of a Hamiltonian necessarily has long-range magic that is parametrically large in system size, and cannot be removed by constant-depth circuits. In this work, we established such a conclusion for the non-Abelian string-net Hamiltonian under an additional restriction: the preparing circuit depth must not scale too unfavorably with the energy density. Because of this condition, our result does not yet prove NLTM in the strongest sense. This limitation is unsurprising, since geometrically-local Hamiltonians cannot satisfy NLTS~\cite{Hastings:2012xdc}, which is necessary for NLTM. This suggests broadening the search beyond conventional Euclidean lattices. Recent constructions of non-Abelian topological order in non-Euclidean geometries~\cite{Zhu:2026vec,Christos:2026yyd,McDonough:2026ngv} provide an intriguing starting point.

A second question concerns how long-range magic depends on the underlying anyon data, particularly beyond the non-chiral topological orders considered here. While our results establish extensive long-range magic in a broad class of non-Abelian phases, they do not yet determine which topological features control its magnitude. It is natural to ask whether this dependence is governed by quantities such as quantum dimensions, fusion multiplicities, or modular data, and whether different non-Abelian phases can be partially ordered by the amount of long-range magic they support. Our present approach is most directly applicable to non-chiral fixed-point phases, where the entanglement-bootstrap framework provides access to fusion spaces through local reduced density matrices, so extending this perspective more broadly remains an important open direction. Clarifying these issues could place long-range magic on a footing similar to that of more familiar universal characteristics of topological order.

%A second question is how long-range magic depends on the underlying anyon data. Our results establish the presence of extensive long-range magic in a broad class of non-Abelian phases, but they do not yet identify which topological characteristics control its magnitude. It is natural to ask whether long-range magic is governed by quantities such as quantum dimensions, fusion multiplicities, or modular data, and whether different non-Abelian phases can be partially ordered according to the amount of long-range magic they support. Clarifying this dependence could help place long-range magic on similar footing to more familiar universal data in topological order. \YZ{Another natural direction is to go beyond the non-chiral topological orders considered here. Our argument applies most directly to non-chiral fixed-point topological orders, where the entanglement-bootstrap framework gives access to fusion spaces through local reduced density matrices.}

A related issue is the operational meaning of the long-range magic detected here. In particular, it would be valuable to understand its relation to \emph{distillable} magic. Our results concern the intrinsic non-stabilizer structure of the state, as quantified by information-theoretic measures, but they do not directly address how much of this magic can be converted into a useful resource for fault-tolerant quantum computation. Establishing a connection between long-range magic and distillation protocols would sharpen the computational significance of our findings.

Another important question is what survives away from exactly solvable fixed-point wavefunctions. How generic is parametrically large long-range magic in many-body quantum states? In Appendix~\ref{sec:LRM_from_mutual_information}, we showed that long-range magic follows quite generally whenever the mutual information obeys an exponentially bounded continuous decay, more precisely under the assumptions of Theorem~\ref{thm:LRM_from_mutual_information}. The proof combines the idea, developed in Refs.~\cite{Parham:2025sxj,Korbany:2025noe}, that non-integer mutual information signals the presence of long-range magic, with the multiplicative fidelity upper bound derived in Appendix~\ref{sec:multiplicative_fidelity_bound}. This suggests that extensive long-range magic may be a rather generic feature of strongly correlated phases, not merely an artifact of fixed-point constructions.

A further open question is the robustness of long-range magic under decoherence and thermal noise. One expects some degree of robustness under weak noise, at least in topological phases where anyonic excitations remain well defined below a threshold~\cite{Fan:2023rvp}. However, the fate of long-range magic at finite temperature or under dissipative dynamics is far from clear. In particular, it is natural to ask whether long-range magic undergoes a sharp transition, and whether such a transition is related to the breakdown of quantum error correction or to known error-correctability thresholds in topological systems.

It would also be interesting to seek an analogue of long-range magic in continuous-variable systems and quantum field theories. Since Gaussian states and free theories are the natural free objects in these settings, long-distance interactions may play a similar role. This suggests a possible link to holographic duality, where semiclassical gravity in AdS is dual to strongly interacting CFTs~\cite{Maldacena:1997re,Gubser:1998bc,Witten:1998qj}. One may speculate whether the emergence of bulk gravity is tied to large long-range magic in the boundary theory~\cite{Cao:2024nrx,Cao:2026uoq}.

Finally, one may ask about the learnability of states with only short-range magic. At present there are two known settings in which efficient learning is possible, both relying on additional structure. In the first, the state is obtained by applying $T$ gates to a stabilizer state; in that case, an efficient learning algorithm exists based on a state hidden subgroup problem~\cite{Lee:2025ldt}. In the second, the state is prepared by applying a constant-depth local unitary to a stabilizer state satisfying the axioms of Ref.~\cite{Kim:2024bsn}. These examples suggest that short-range magic may be tractable when accompanied by sufficient algebraic or geometric control. By contrast, it is tempting to speculate that states with genuinely long-range magic are generically much harder to learn, providing a possible complexity-theoretic interpretation of the phenomenon.

\section*{Acknowledgments}

We thank Jeongwan Haah, Tyler D. Ellison, Nathanan Tantivasadakarn, Zi-Wen Liu, Yingfei Gu, Sarang Gopalakrishnan, Ruochen Ma and Bowen Shi for discussions. YZ and SV acknowledge support from the National Science Foundation under Grant No. DMR–2441671. IK acknowledges support from NSF under award number PHY-2337931. This research was supported in part by grant NSF PHY-2309135 to the Kavli Institute for Theoretical Physics (KITP). 
Y.B. is supported in part by the Gordon and Betty Moore Foundation Grant No. GBMF7392 to the Kavli Institute for Theoretical Physics (KITP).\\

\emph{Note added:} During completion of this work, we became aware of forthcoming work~\cite{Korbany:2026toappear}, which studies long-range non-stabilizerness of topologically encoded states on a torus geometry, focusing on certain string-net models. They introduce a method, based on the evaluation of the mutual information, to diagnose long-range non-stabilizerness. %Our work employs different methods, identifying a local obstruction to stabilizer realizability and amplifying it to extensive long-range magic. %We also show extensive long-range magic in low-energy non-Abelian states and in Abelian string-net states with unfavorable qudit dimensions.

\appendix

\section{Re-phasing stabilizer generators by Pauli conjugation}\label{sec:generators_rephased_by_paulis}
In this appendix we show that in general qudit dimensions, the stabilizer generators can be re-phased by conjugating with Pauli operators. Lemma~\ref{thm:EB_pauli_connects_extreme_points} in the main text then follows.

\begin{lemma}[Re-phasing stabilizer generators by Pauli conjugation]
Let $S$ be a stabilizer group, and choose an independent generating set
\begin{equation}
S=\langle g_1,\dots,g_k\rangle ,
\end{equation}
where $g_i$ has order $\delta_i$ modulo phases. Let $\zeta_i=e^{2\pi i/\delta_i}$. Then for any choice of integers $u_i\in\mathbb Z_{\delta_i}$, there exists a Pauli operator $P$ such that
\begin{equation}
P g_i P^\dagger=\zeta_i^{u_i}g_i,
\qquad i=1,\dots,k.
\end{equation}
\end{lemma}
\begin{proof}
Let $\calP$ denote the Pauli group modulo phases. For every Pauli operator $P\in\calP$, define $\chi_P:S\to U(1)$ by the commutation relation
\begin{equation}
PQ=\chi_P(Q)QP,\qquad Q\in S.
\end{equation}
Since commutation phases multiply under products of elements in $S$, $\chi_P$ is a character of $S$. Explicitly, if $P=\otimes_iZ^{a_i}X^{b_i}$ and $Q=\otimes_iZ^{c_i}X^{d_i}$, then $\chi_P(Q)=\omega^{\sum_i(a_id_i-b_ic_i)}$. 

Let $\widehat{S}$ be the character group of $S$. Our goal is to show that for every character in $\chi\in \widehat{S}$, there exists a Pauli $P$ such that $\chi=\chi_P$, that is, the map 
\begin{equation}
\calP\longrightarrow \widehat{S},\qquad P\mapsto \chi_P
\end{equation}
is onto. The map is not injective: let $L$ be the subgroup of $\calP$ that commutes with $S$, then for $\forall P'\in L$, $\chi_P=\chi_{PP'}$. Conversely, if $\chi_P=\chi_{P'}$, then $\chi_{P^{-1}P'}=1$ on all of $S$, so $P^{-1}P'\in L$. Thus the distinct characters obtained from Pauli operators are labeled by the quotient $\calP/L$. Therefore, to show that $P\mapsto\chi_P$ is onto, it suffices to show
\begin{equation}
|\calP/L|=|\widehat{S}|.
\end{equation}

Since $S$ is a finite Abelian group, its character group has the same size: $|\widehat{S}|=|S|$. It remains to compute $|L|$. We use the character-orthogonality identity
\begin{equation}
\frac{1}{|S|}\sum_{Q\in S}\chi_P(Q)=\begin{cases}1,& P\in L\\0,&\text{otherwise}\end{cases}
\end{equation}
Summing over $P$ gives
\begin{equation}
|L|=\frac{1}{|S|}\sum_{Q\in S}\sum_{P\in \calP}\chi_P(Q).
\end{equation}
Now evaluate the sum over $P\in\calP$. If $Q\neq I$, then there must exists Pauli's that doesn't commute with $Q$, because if $Q=\otimes_i Z^{a_i}X^{b_i}$ and $a_i\neq 0$, then $Q$ doesn't commute with $X_i$; if $b_i\neq 0$, then $Q$ doesn't commute with $Z_i$. Thus $\sum_{P\in \calP}\chi_P(Q)=|\calP|\delta_{Q,I}$, and
\begin{equation}
|L|=\frac{|\calP|}{|S|}.
\end{equation}
Consequently, $|\calP/L|=|\calP|/|L|=|\widehat{S}|$. Therefore the image of the map $P\mapsto\chi_P$ has the same size as $\widehat{S}$, and hence the map is onto. Thus for every character $\chi\in\widehat{S}$, there exists a Pauli operator $P$ such that $\chi=\chi_P$.
\end{proof}

\section{Upper bound on magic in general qudit dimensions}\label{sec:Rom_upper_bound}
In this appendix we review the upper bound on robustness of magic from Ref.~\cite{Liu:2020yso} and extend the argument to arbitrary composite qudit dimensions. The additional ingredient is a cover of the composite-dimensional Pauli group by maximal commuting Pauli subgroups, given in Lemma~\ref{thm:cover_pauli_group} below.

\begin{proposition}\label{thm:Rom_upper_bound}
Consider a system of $n$ qudits of local dimension $q$. For an arbitrary state, the log-robustness of magic is upper bounded by
\begin{equation}
\Lr(\rho)\le (n+2^{-n-1})\log q.
\end{equation}
\end{proposition}
\begin{proof}
Let $\calS$ denote the convex hull of pure stabilizer states. Recall that the robustness of magic $\R(\rho)$ and its logarithmic version $\LR(\rho)$ are defined as
\begin{equation}\begin{aligned}
\Lr(\rho)&=\log[1+2\R(\rho)],\; \operatorname{R}(\rho):=\min\{s\ge0:\exists\sigma,\sigma'\in\calS\ \st\ \rho=(1+s)\sigma'-s\sigma\}, 
\end{aligned}\end{equation}
Equivalently,
\begin{equation}
1+2\R(\rho)=\min\bigg\{\sum_\sigma|c_\sigma|\ \st\ \rho=\sum_{\sigma\in\calS}c_\sigma\sigma,\ c_\sigma\in\mathbb{R}\bigg\}.
\end{equation}
This is a linear program, and its dual program is
\begin{equation}
1+2\R(\rho)=\max_{A}\bigg\{|\Tr(A\rho)|\ \st\ |\Tr(A\sigma)|\le 1,\ \forall\sigma\in\calS\bigg\}.
\end{equation}

In the next step, we give an upper bound to $\lVert A\rVert_2$, which upper bounds $|\Tr(A\rho)|$. Let $\calP$ denote the Pauli group modulo phases. Expand $A$ in the Pauli basis:
\begin{equation}
A=\sum_{P\in\calP}c_PP,\qquad c_P=q^{-n}\Tr(P^\dagger A).
\end{equation}
Then $\lVert A\rVert_2^2=q^n\sum_{P\in\calP}|c_P|^2$. Let $S\subset\calP$ be the stabilizer group of a pure stabilizer state. Thus $S$ is a maximal commuting Pauli subgroup of size $|S|=q^n$. One can define a family of pure stabilizer states (corresponding to re-phasing the stabilizer generators) using all $\chi\in\widehat{S}$, where $\widehat{S}$ is the character group of $S$:
\begin{equation}
\sigma_{S,\chi}:=q^{-n}\sum_{P\in S}\chi(P)P.
\end{equation}
The condition of $|\Tr(A\sigma_{S,\chi})|\le 1$ imposes
\begin{equation}
\bigg|\sum_{P\in S}\chi^*(P)c_P\bigg|\le 1,\quad\forall\chi\in\widehat{S}.
\end{equation}
Since the character table $\big[q^{-n/2}\chi(P)\big]_{P\in S,\chi\in\widehat{S}}$ is unitary for any finite Abelian group, we have
\begin{equation}
\sum_{P\in S}|c_P|^2\le q^{-n}\sum_{\chi\in\widehat{S}}\bigg|\sum_{P\in S}\chi^*(P)c_P\bigg|^2\le q^{-n}|\widehat{S}|=1,
\end{equation}
where we used $|\widehat{S}|=|S|=q^n$. Let $q=\prod_jp_j^{r_j}$, where $p_j$ are distinct primes. By Lemma~\ref{thm:cover_pauli_group}, there exist $q^n\prod_j(1+p_j^{-n})$ maximal stabilizer groups whose union covers $\calP$. Summing over these stabilizer groups give
\begin{equation}
\sum_{P\in\calP}|c_P|^2\le q^n\prod_j(1+p_j^{-n}).
\end{equation}
And thus
\begin{equation}
|\Tr(A\rho)|\le\lVert A\rVert_\infty\le\lVert A\rVert_2\le q^n\prod_j\sqrt{1+p_j^{-n}}.
\end{equation}
Taking the logarithm,
\begin{equation}
\Lr(\rho)\le n\log q+\frac{1}{2}\sum_j\log(1+p_j^{-n})\le n\log q+\frac{1}{2}\sum_j 2^{-n}\le (n+2^{-n-1})\log q.
\end{equation}
\end{proof}

We now prove the Pauli covering lemma used above. This generalizes the standard construction for prime-dimensional qudits~\cite{Bandyopadhyay:2001zqr} to arbitrary composite qudit dimension.

\begin{lemma}\label{thm:cover_pauli_group}
In an $n$-qudit system, let $q=\prod_jp_j^{r_j}$ be the qudit dimension, where $p_j$ are distinct primes. There exist a collection of $q^n\prod_j(1+p_j^{-n})$ maximal stabilizer groups, each corresponding to a pure stabilizer state, whose union covers $\calP$.
\end{lemma}
\begin{proof}
In composite qudit dimensions, there is a canonical local basis transformation that maps the Pauli operators to a tensor product of Pauli operators on prime-power qudits~\cite{Looi:2011jrm} based on the Chinese remainder theorem. Specifically, there is a unitary $\calU$ acting on a single $q$-dimensional qudit such that $\calU^{\otimes n}\calP\calU^{\dagger\otimes n}=\otimes_j\calP_j$, where $\calP_j$ is the Pauli group of qudits of dimension $p_j^{r_j}$. A stabilizer group $S$ gets mapped to a tensor product of stabilizer groups over these lower-dimensional qudits. Thus it suffices to construct, for each prime-power dimension $p_j^{r_j}$, a cover of $\calP_j$ by $p_j^{nr_j}+p_j^{n(r_j-1)}$ maximal stabilizer groups, and then take the tensor product of such groups.

We therefore fix a prime power $p^r$ and work on $n$ qudits of dimension $p^r$. Pauli operators are labeled by 
\begin{equation}
P=Z^{\bolda}X^{\boldb}:=\prod_{i=1}^n Z^{a_i}X^{b_i},\qquad (\bolda, \boldb)\in\mathbb{Z}_{p^r}^{2n}.
\end{equation}
The commutation relation between two Pauli's $P=Z^{\bolda}X^{\boldb}$, $Q=Z^{\boldc}X^{\boldd}$ is given by $PQ=\omega^{\bolda\cdot\boldd-\boldb\cdot\boldc}QP$, where $\bolda\cdot\boldd-\boldb\cdot\boldc$ is the symplectic inner product between $(\bolda,\boldb)$ and $(\boldc,\boldd)$. Thus any additive subgroup of $\mathbb Z_{p^r}^{2n}$ of size $p^{nr}$
with vanishing pairwise symplectic form defines a maximal commuting
Pauli subgroup. Now the problem reduces to construct a collection of such sets that covers the entire $\mathbb{Z}_{p^r}^{2n}$.

To this end, we map $\mathbb{Z}_{p^r}^{n}$ to the Galois ring $R:=GR(p^r,n)$. Take a basis $\{e_1,\cdots,e_n\}$ of $R$ and take the dual basis $\{e_1^*,\cdots,e_n^*\}$ satisfying $\Tr(e_ie_j^*)=\delta_{ij}$, where $\Tr$ is the ring trace~\cite{sison2014basesgaloisringgrprm,wan2003lectures} that maps $R$ to $Z_{p^r}$. Such a dual basis can be explicitly found in~\cite{sison2014basesgaloisringgrprm}. Define the maps
\begin{equation}
f,g:\mathbb{Z}_{p^r}^{n}\rightarrow R,\qquad f(\bolda)=\sum_ia_ie_i,\quad g(\boldb)=\sum_ib_ie_i^*.
\end{equation}
By linearity of the trace:
\begin{equation}
\bolda\cdot\boldb=\Tr[f(\bolda)g(\boldb)].
\end{equation}
Thus, a maximal set in $\mathbb{Z}_{p^r}^{2n}$ with vanishing mutual symplectic inner products corresponds to a set of $(x,y)\in R\times R$ with $p^{nr}$ elements, such that any two elements $(x,y)$, $(z,w)$ satisfy $\Tr(xw-yz)=0$. We need to construct a collection of these sets whose union cover $R\times R$. The construction is as follows. Consider the sets
\begin{equation}\begin{aligned}
E_t&:=\{(x,tx):x\in R\}\quad\text{for}\quad t\in R,\\
F_s&:=\{(sy,y):y\in R\}\quad\text{for}\quad s\in pR.
\end{aligned}\end{equation}
parametrized by $t\in R$ and $s\in pR$. These sets are additive, and each set has size $|R|=p^{nr}$ as desired. Moreover, if $(x,tx),(z,tz)\in E_t$, then
\begin{equation}
\Tr(x(tz)-(tx)z)=0,
\end{equation}
and similarly for $F_s$.

There are $|R|+|pR|=p^{nr}+p^{n(r-1)}$ sets in total. We claim that they cover $R\times R$. For any $(x,y)$, take the largest $k$ such that $x,y\in p^kR$. Thus writing $(x,y)=p^k(x_0,y_0)$, either $x_0$ or $y_0$ is not in $pR$. Not being in $pR$ means that the element is a unit---it has an inverse. If $x_0\notin pR$, then one can take $t=y_0x_0^{-1}$ and thus $(x,y)\in E_t$. Instead if $x_0\in pR$ and $y_0\notin pR$, then one can take $s=x_0y_0^{-1}$ and find $(x,y)\in F_s$. Mapping these sets back to Pauli labels gives a cover of the Pauli group by $p^{nr}+p^{n(r-1)}$ maximal stabilizer groups in dimension $p^r$. Tensoring the covers over the distinct prime-power factors of $q$ proves the lemma.

\end{proof}

\section{Multiplicative fidelity upper bound with stabilizer projection states}\label{sec:multiplicative_fidelity_bound}

In this appendix we show that the fidelity between a product state and a stabilizer projection state is exponentially small in the number of components, provided that the distance between each component and stabilizer projection states is uniformly lower bounded. This result is tailored to the pure state argument in the main text: after tracing a pure stabilizer state to disjoint patches, one obtains a stabilizer projection state, which may still be entangled across the patches. The bound below controls precisely this situation.

A complementary question is whether local magic on many disjoint subsystems implies extensive magic for a general, not necessarily product, many-body state. We address this more general formulation in Appendix~\ref{sec:extensive_magic_from_magical_subsystems}.

We make use of the notion of distinguishability under restricted measurements~\cite{Matthews:2008tci}. In particular, we consider stabilizer measurements ($\SM$), POVMs where every element is proportional to some stabilizer projection state. For example, Pauli measurements are stabilizer measurements, because the spectral projectors of a Pauli operator defines stabilizer projection states (Definition~\ref{def:stabilizer_projection_states}).

We define the $\SM$-fidelity by
\begin{equation}
F_{\SM}(\rho,\sigma)=\min_{\calM\in\SM}F\big(\calM(\rho),\calM(\sigma)\big),\qquad\calM(\cdot):=\sum_x\Tr[M_x(\cdot)]\ketbra{x},
\end{equation}
Equivalently, the $\SM$-fidelity is the smallest classical fidelity between the outcome distributions of $\rho$ and $\sigma$ that can be achieved using stabilizer measurements. Thus $F_{\SM}(\rho,\sigma)$ measures how well $\rho$ and $\sigma$ can be distinguished by stabilizer measurements. Since restricting the allowed measurements only makes two states harder to distinguish, the $\SM$-fidelity gives an upper bound to fidelity:
\begin{equation}
F(\rho,\sigma)\le F_{\SM}(\rho,\sigma),
\end{equation}
which follows from the monotonicity of fidelity under quantum channels.

For our purposes, we need the corresponding notions of closeness to the set $\SP$ of stabilizer projection states. We define
\begin{equation}
F(\rho,\SP):=\max_{\tau\in\SP}F(\rho,\tau),
\end{equation}
and similarly
\begin{equation}
F_{\SM}(\rho,\SP):=\max_{\tau\in\SP}F_{\SM}(\rho,\tau).
\end{equation}
The monotonicity inequality above immediately implies
\begin{equation}
F(\rho,\SP)\le F_{\SM}(\rho,\SP).
\end{equation}

The key ingredient for showing extensive magic in the main text is to give a multiplicative fidelity upper bound on the fidelity between a product state and a stabilizer projection state. The advantage of $F_{\SM}$ is that it interacts well with stabilizer measurements and allows a multiplicative upper bound even when the stabilizer projection state used for comparison is entangled across different subsystems. We start with the case of two subsystems and the product state $\rho_1\otimes\rho_2$. Its distance from stabilizer projection states $\SP_{12}$ can be lower bounded by the following Lemma based on~\cite{Lancien_2017}.
\begin{lemma}\label{thm:fidelity_two_party}
\begin{equation}
F(\rho_1\otimes\rho_2,\SP_{12})\le F_{\SM}(\rho_1,\SP_1)F(\rho_2,\SP_2)
\end{equation}    
\end{lemma}
\begin{proof}
Choose $\sigma_{12}\in\SP_{12}$ such that $F(\rho_1\otimes\rho_2,\SP_{12})=F(\rho_1\otimes\rho_2,\sigma_{12})$. Denote by $\sigma_1=\Tr_2(\sigma_{12})$ and $\sigma_2=\Tr_1(\sigma_{12})$ the reduced density matrices on the respective subsystems. Perform an arbitrary stabilizer measurement on the first system, represented by a channel $\calM\otimes\calI$. For each measurement outcome $x$, define $\sigma_{2,x}=\frac{\Tr_1(M_x\sigma_{12})}{\Tr(M_x\sigma_1)}$, which can be understood as the post-measurement state on the second subsystem conditioned on $x$. The output of the channel can be written as:
\begin{equation}
\begin{aligned}
\calM\otimes\calI(\rho_1\otimes\rho_2)&=\sum_x\Tr(M_x\rho_1)\ketbra{x}\otimes \rho_2,\quad\calM\otimes\calI(\sigma_{12})=\sum_x\Tr(M_x\sigma_1)\ketbra{x}\otimes \sigma_{2,x}.
\end{aligned}
\end{equation}
Since fidelity is monotonic under channels, we have
\begin{equation}
\begin{aligned}
F(\rho_1\otimes\rho_2,\sigma_{12})&\le F\big(\calM\otimes\calI(\rho_1\otimes\rho_2),\calM\otimes\calI(\sigma_{12})\big) \\
&=\sum_x\sqrt{\Tr(M_x\rho_1)\Tr(M_x\sigma_1)}F(\rho_2,\sigma_{2,x}) \\
&\le F\big(\calM(\rho_1),\calM(\sigma_1)\big)F(\rho_2,\SP_2)\le\max_{\sigma_1\in\SP_1}F\big(\calM(\rho_1),\calM(\sigma_1)\big)F(\rho_2,\SP_2)
\end{aligned}
\end{equation}
The third line uses $\sigma_{2,x}\in\SP_2$ is still a stabilizer projection state, as $\sigma_{12}\in\SP_{12}$ and we are just doing stabilizer measurements. Minimizing over $\calM\in\SM$ gives the desired inequality.
\end{proof}

\begin{proposition}[Multiplicative fidelity upper bound for product states]\label{thm:multiplicative_fidelity_bound}
Take a product state over multiple subsystems $\otimes_i\rho_i$. Let $\SP_i$ denote the stabilizer projection states on subsystem $i$, then
\begin{equation}
F(\otimes_i\rho_i,\SP)\le\prod_iF_{\SM}(\rho_i,\SP_i)
\end{equation}
\end{proposition}
\begin{proof}
Iterate lemma~\ref{thm:fidelity_two_party} and note $F(\rho,\SP)\le F_{\SM}(\rho,\SP)$.
\end{proof}

Now it remains to control $F_{\SM}(\rho,\SP)$ in terms of the trace distance to $\SP$. For this purpose we introduce the stabilizer measurement version of trace disntance:
\begin{equation}
\|\rho-\sigma\|_{\SM}:=\max_{\calM\in\SM}\lVert\calM(\rho)-\calM(\sigma)\rVert_1.
\end{equation} 
To control $\|\rho-\sigma\|_{\SM}$ using $\|\rho-\sigma\|_1$ is really asking how good stabilizer measurements can distinguish between two distinct states. Since Pauli measurements are stabilizer measurements that are informationally complete, we can construct a Pauli measurement to successfully distinguish between arbitrary distinct states, at a cost of a dimension factor.

\begin{lemma}[Distinguishability of stabilizer measurements]\label{thm:SM_distinguishability}
Let $\mathrm{D}$ be the Hilbert space dimension.
\begin{equation}
\norm{\rho-\sigma}_{\SM}\ge\frac{1}{\mathrm{D}}\onenorm{\rho-\sigma}.
\end{equation}
\end{lemma}
\begin{proof}
Let $\Delta:=\rho-\sigma$. Expand it under the basis of Pauli operators:
\begin{equation}
\Delta=\frac{1}{\mathrm{D}}\sum_{P}\alpha_P P,\quad\text{where}\quad\alpha_P=\Tr(P^\dagger\Delta).
\end{equation}
The Pauli operators satisfy $\Tr(P^\dagger Q)=\mathrm{D}\delta_{P,Q}$, so
\begin{equation}
\lVert\Delta\rVert_2^2=\frac{1}{\mathrm{D}}\sum_P|\alpha_P|^2.
\end{equation}
Since there are $\mathrm{D}^2$ different Pauli's in total, there must exist a Pauli $P$ with $|\alpha_P|\ge\frac{1}{\sqrt{\mathrm{D}}}\lVert\Delta\rVert_2$. Using $\lVert\Delta\rVert_1\le\sqrt{\mathrm{D}}\lVert\Delta\rVert_2$, we obtain
\begin{equation}
|\alpha_P|\ge\frac{1}{\mathrm{D}}\lVert\Delta\rVert_1.
\end{equation}

Consider the spectral measurement of the Pauli $P=\sum_x\lambda_x M_x$, where $M_x$ are projectors and $\lambda_x$ are phases. Under the channel $\calM(\cdot)=\sum_x\Tr[M_x(\cdot)]\ketbra{x}$, the distance becomes
\begin{equation}
\lVert\calM(\rho)-\calM(\sigma)\rVert_1=\sum_x|\Tr(M_x\Delta)|\ge\left|\sum_x\lambda_x\Tr(M_x\Delta)\right|=|\alpha_P|\ge\frac{1}{\mathrm{D}}\lVert\Delta\rVert_1.
\end{equation}
Since $\lVert\rho-\sigma\rVert_{\SM}\ge\lVert\calM(\rho)-\calM(\sigma)\rVert_1$, the desired bound is established.

\end{proof}

\begin{proposition}\label{thm:exp_small_fidelity_bound}
Given a tensor product state $\otimes_{i=1}^m\rho_i$ over $m$ disjoint subregions, and each region has Hilbert dimension $\mathrm{D}$, if $\onenorm{\rho_i-\SP}\ge\epsilon$, $\forall i$, then
\begin{equation}    
F(\otimes_{i=1}^m\rho_i,\SP)\le \left(1-\frac{\epsilon^2}{4\mathrm D^2}\right)^\frac{m}{2}.
\end{equation}
\end{proposition}
\begin{proof}
Fix $i$ and $\tau\in\SP_i$. By Lemma~\ref{thm:SM_distinguishability}, there exists a stabilizer measurement $\calM$ such that
\begin{equation}
\bigl\|\calM(\rho_i)-\calM(\tau)\bigr\|_1
\ge \frac{1}{\mathrm D}\|\rho_i-\tau\|_1
\ge \frac{\epsilon}{\mathrm D}.
\end{equation}
Applying the Fuchs--van de Graaf inequality to the states $\calM(\rho_i)$ and $\calM(\tau)$ gives
\begin{equation}
F\bigl(\calM(\rho_i),\calM(\tau)\bigr)
\le
\sqrt{1-\frac{\epsilon^2}{4\mathrm D^2}}.
\end{equation}
Since $F_{\SM}(\rho_i,\tau)$ is the minimum over all stabilizer measurements,
\begin{equation}
F_{\SM}(\rho_i,\tau)
\le
\sqrt{1-\frac{\epsilon^2}{4\mathrm D^2}}.
\end{equation}
Maximizing over $\tau\in\SP_i$ yields
\begin{equation}
F_{\SM}(\rho_i,\SP_i)
\le
\sqrt{1-\frac{\epsilon^2}{4\mathrm D^2}}.
\end{equation}
The claim now follows from Proposition~\ref{thm:multiplicative_fidelity_bound}.
\end{proof}

\section{From magical subsystems to extensive magic}\label{sec:extensive_magic_from_magical_subsystems}
In this appendix we give extensive lower bounds to states with extensive number of subregions with magic, which may be of independent interest. It has been generally difficult to prove extensiveness of magic monotones for generic product states, except for odd-prime dimensional qudits where mana is an extensive monotone. In other situations, it is only proved for single qubit states and a restricted set of few qubit states\cite{Bravyi:2018ugg,Seddon:2021rij,Saxena:2022usu,Rubboli:2023qzz}.

For a product state, we can directly generalize Theorem~\ref{thm:fidelity_two_party}, replacing $\SP$ with $\calS$ in the proof:

\begin{proposition}[Extensive log-stabilizer fidelity in tensor product of magic states]
Given a product state $\otimes_{i=1}^n\rho_i$, we have 
\begin{equation}
F(\otimes_i\rho_i,\calS)\le\prod_iF_{\SM}(\rho_i,\calS_i),
\end{equation}
where
\begin{equation}
F_{\SM}(\rho,\sigma):=\min_{\calM\in\SM}F\big(\calM(\rho),\calM(\sigma)\big),\qquad F_{\SM}(\rho,\calS):=\max_{\tau\in\calS}F_{\SM}(\rho,\tau).
\end{equation}
Assuming $\onenorm{\rho_i-\calS}\ge\epsilon$, $\forall i$, then
\begin{equation}
\LF(\otimes_{i=1}^n\rho_i)\ge n\log\left(1-\frac{\epsilon^2}{4\mathrm D^2}\right)^{-1}.
\end{equation}
\end{proposition}

Even without a tensor product structure, one can show that the state has extensive magic measured by the relative entropy of magic, following~\cite{Piani:2009pzy}. 
\begin{proposition}[Extensive relative entropy of magic from local magic]
Given any state $\rho$ and any set of disjoint subsystems $A_1,\cdots, A_n$, we have
\begin{equation}
S(\rho\|\calS)\ge\sum_i S_{\SM}(\rho_i\|\calS_i),
\end{equation}
where
\begin{equation}
S_{\SM}(\rho\|\sigma):=\max_{\calM\in\SM}S\big(\calM(\rho)\|\calM(\sigma)\big),\qquad S_{\SM}(\rho\|\calS):=\min_{\tau\in\calS}S_{\SM}(\rho\|\tau).
\end{equation}
Assuming $\onenorm{\rho_i-\calS}\ge\epsilon$, $\forall i$, then $S(\rho\|\calS)\ge \frac{n\epsilon^2}{2\mathrm{D}^2}$. 
\end{proposition}
\begin{proof}
We start with the case with two subsystems. For $\forall \calM\in\SM$ and $\sigma_{12}\in\calS$,
\begin{equation}
\begin{aligned}
S(\rho_{12}\|\sigma_{12})
&\ge S\big(\calM\otimes\calI(\rho_{12})\|\calM\otimes\calI(\sigma_{12})\big) \\
&=S\big(\calM(\rho_1)\|\calM(\sigma_1)\big)+\sum_xp_xS(\rho_{2,x}\|\sigma_{2,x}),\qquad p_x=\Tr(M_x\rho_1),\ q_x=\Tr(M_x\sigma_1) \\
&\ge S\big(\calM(\rho_1)\|\calM(\sigma_1)\big)+S\Big(\rho_2\Big\|\sum_xp_x\sigma_{2,x}\Big)\\
&\ge S\big(\calM(\rho_1)\|\calM(\sigma_1)\big)+S(\rho_2\|\calS_2)
\end{aligned}
\end{equation}
The third line is from the joint convexity of relative entropy\footnote{This follows from monotonicity under quantum channels, as one can start from the classical-quantum states $\sum_xp_x\ketbra{x}\otimes\rho_x$ and $\sum_xp_x\ketbra{x}\otimes\sigma_x$, whose relative entropy is $\sum_xp_xS(\rho_x\|\sigma_x)$ and then take the partial trace.}. The fourth line is because $\sum_xp_x\sigma_{2,x}\in\calS$. This implies
\begin{equation}
S(\rho_{12}\|\calS_{12})\ge S_{\SM}(\rho_1\|\calS_1)+S(\rho_2\|\calS_2).
\end{equation}
Iterating this procedure, we get the desired inequality for multiple subsystems. Finally, plug in the quantum Pinsker inequality to convert trace distance to relative entropy.
\end{proof}

\section{Logarithmically large long-range magic from continuous decay of mutual information}\label{sec:LRM_from_mutual_information}
Entanglement entropies in stabilizer states are discrete: the mutual information between any two subsystems must be an integer multiple of $\log q$, when $q$ is a prime number. For composite qudit dimensions, the mutual information cannot sit between $0$ and $\log p$, where $p$ is the smallest prime divisor of $q$. It suggests that stabilizer states cannot represent states with continuously decaying mutual information---a property that is stable under constant-depth local unitaries. Mutual-information-based ideas of this kind have already been used to prove the \emph{existence} of long-range magic in certain states~\cite{Korbany:2025noe,Parham:2025sxj}. Here we show that there must be \emph{logarithmically large} long-range magic, under the condition of continuous decay of mutual information. The argument parallels Sec.~\ref{subsec:LRM_nA}: first, a constant-size patch whose internal mutual information falls in an interval forbidden for stabilizer projection states is uniformly bounded away from $\SP$; second, exponential decay of mutual information makes the state on many well-separated patches close to a product state; and third, Proposition~\ref{thm:exp_small_fidelity_bound} amplifies this local obstruction into an $\Omega(\log n)$ lower bound. This criterion shows that large long-range magic is not tied exclusively to non-Abelian topological order. %In particular, generic one-dimensional MPS's are expected to satisfy such conditions, closely related to the MPS fixed-point criterion for long-range magic developed in Ref.~\cite{Korbany:2025noe}.}

For a region $A$, let $A^{+d}$ denote the region within distance $d$ to $A$, and let $A^{-d}$ denote the the region with distance $d$ away from the complement of $A$.
\begin{theorem}\label{thm:LRM_from_mutual_information}
Let $\rho=\ketbra{\psi}$ be a pure state on a $n$ qudits, and let $q$ be the qudit dimensions, whose smallest prime divisor is $p$. For a fixed circuit depth $d=O(1)$. Assume that the mutual information on the state $\rho$ satisfies the following properties:
\begin{enumerate}
\item For any two disjoint subsystems $A$ and $B$, the mutual information is upper bounded by $I_\rho(A:B)\le K|A|\cdot|B|\exp(-\dist(A:B)/\xi)$ with $\xi=O(1)$. 

\item There exist $m$ disjoint subregions $A_1,\cdots, A_m$ and constants $c_0,c_1,r_0,\epsilon_1,\epsilon_2>0$, independent of $n$, such that $m\ge c_0\log n$, $|A_i|\le r_0$, $\dist(A_i,A_j)\ge c_1\log n\ (i\neq j)$ and within each $A_i$, there exists two subregions $C_i$ and $D_i$ satisfying $\epsilon_1\le I_\rho(C^{-d}:D^{-d})\le I_\rho(C^{+d}:D^{+d})\le\epsilon_2$, where $0<\epsilon_1\le\epsilon_2<\log p$, $C^{+d}$ and $D^{+d}$ are disjoint.
\end{enumerate}
Then
\begin{equation}
\LR M_d(\ketbra{\psi})=\Omega(\log n)\quad\text{for}\quad d=O(1).
\end{equation}
\end{theorem}

\begin{proof}
For a depth-$d$ unitary, we write $\tilde{\rho}=U\ketbra{\psi}U^\dagger$. On the evolved state $\tilde{\rho}$, take the core of the patches $A_1^{-d},\cdots,A_m^{-d}$. Because the original patches are separated by distance at least $c_1\log n$, the regions $A_i^{-d}$ are disjoint. 

We first show that the state on the union of patches $A^{-d}=\cup_{i=1}^m A_i^{-d}$ is close to a tensor product state, due to small mutual information between the patches. Consider the relative entropy between the the tensor product of state over each patch $\otimes_{i=1}^m\tilde{\rho}_i$ and the global state $\tilde{\rho}_{A^{-d}}$ over all patches.
\begin{equation}
S(\tilde{\rho}_{A^{-d}}\|\otimes_{i=1}^m\tilde{\rho}_i)=\sum_iS_{\tilde{\rho}}(A_i)-S_{\tilde{\rho}}(A)=\sum_{i=1}^{m-1}I_{\tilde{\rho}}\left(\cup_{j=1}^iA_j^-:A_{i+1}^-\right).
\end{equation}
By Lemma~\ref{thm:stability_of_MI}, each term is upper bounded by the mutual information of the corresponding thickened regions in the original state:
\begin{equation}
I_{\tilde\rho}\!\left(A_1^{-d}\cup\cdots\cup A_i^{-d}:A_{i+1}^{-d}\right)
\le
I_{\rho}\!\left(A_1\cup\cdots\cup A_i:A_{i+1}\right).
\end{equation}
Using Assumption~1 together with $|A_i|\le r_0$ and $\dist(A_i,A_j)\ge c_1\log n$, we obtain
\begin{equation}
S(\tilde{\rho}_{A^{-d}}\|\otimes_{i=1}^m\tilde{\rho}_i)\le K(c_0\log n)^2r_0^2 n^{-c_1/\xi},
\end{equation}
where we plugged $m=c_0\log n$. To lower bound fidelity, recall that the $\alpha$-R\'enyi relative entropy is monotonically increasing in $\alpha$~\cite{Muller-Lennert:2013liu,Beigi:2013teh}, and the case $\alpha=\tfrac12$ corresponds to fidelity. Hence $F(\rho_{A^-},\otimes_{i=1}^m\rho_i)\ge\exp(-S(\rho_{A^-}\|\otimes_{i=1}^m\rho_i)/2)\ge 1-1/\poly(n)$.

We next show that each local state $\rho_i$ is uniformly bounded away from $\SP$. By Assumption~2 and Lemma~\ref{thm:stability_of_MI}, for every $i$ we have
\begin{equation}
\epsilon_1\le I_{\tilde\rho}(C_i:D_i)\le\epsilon_2<\log p.
\end{equation}
If $\rho_i$ were a stabilizer projection state, then its reduction to $C_i\cup D_i$ would also be a stabilizer projection state, and therefore the mutual information $I_{\tilde\rho}(C_i:D_i)$ would have to be an integer multiple of $\log q$, a contradiction. Thus $\rho_i\notin\SP$ for every $i$. Since each patch has constant size, there exists a constant $\epsilon>0$, independent of $n$, such that
\begin{equation}
\min_{\sigma\in\SP}\|\rho_i-\sigma\|_1\ge \epsilon,
\quad \forall i.
\end{equation}

From the multiplicative fidelity upper bound in Proposition~\ref{thm:exp_small_fidelity_bound}, we have $F(\otimes_{i=1}^m\rho_i,\SP)\le 1/\poly(n)$. To relate this to $F(\rho_A,\SP)$, we use Lemma~\ref{thm:fidelity_triangle_inequality}, which implies $F(\rho_A,\SP)\le 1/\poly(n)$. Now let $\ket{\phi}$ be a pure stabilizer state achieving the maximum in $F(\tilde\rho,\calS)$. Under partial trace to $A^-$, the pure stabilizer state $\ketbra{\phi}$ reduces to some stabilizer projection state $\sigma\in\SP$. By monotonicity of fidelity under partial trace,
\begin{equation}
F(\tilde\rho,\calS)=F(\tilde\rho,\ketbra{\phi})\le F(\tilde\rho_{A^-},\sigma)\le F\!\left(\tilde\rho_{A^-},\SP\right) \le 1/\poly(n).
\end{equation}
Therefore
\begin{equation}
\LF(\tilde\rho)=\Omega(\log n).
\end{equation}
Because the choice of $U$ was arbitrary, this implies
\begin{equation}
\LR M_d(\ketbra{\psi})=\Omega(\log n)\quad\text{for}\quad d=O(1).
\end{equation}
\end{proof}

The following lemma makes precise the stability of mutual information under constant-depth local unitaries; see, for example, Ref.~\cite{Parham:2025sxj}.
\begin{lemma}\label{thm:stability_of_MI}
Take any state $\rho$ evolved by a local unitary with depth $d$. For any two disjoint regions $A$ and $B$ where $A^{+d}$ and $B^{+d}$ are disjoint, we have
\begin{equation}
I_\rho(A^{-d}:B^{-d})\le I_{U\rho U^\dagger}(A,B)\le I_\rho(A^{+d}:B^{+d}).
\end{equation}
\end{lemma}
\begin{proof}
After applying the finite depth local unitary $U$ to the initial state $\rho$, tracing over the complement region of $A\cup B$ cancels all the gates outside the lightcones of $A$ and $B$---these are the unshaded gates below:
\begin{equation*}
\centering\includegraphics[width=0.7\linewidth]{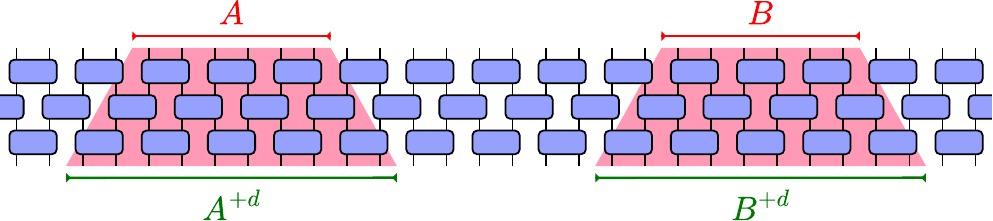}
\end{equation*}
Let $C$ be the complement region of $A\cup B$. The density matrix over $A\cup B$ is therefore 
\begin{equation}
\Tr_{C}(U\rho U^\dagger)=\Tr_{C}[(U_{A^{+d}}\otimes U_{B^{+d}})\rho(U_{A^{+d}}\otimes U_{B^{+d}})^\dagger],
\end{equation}
where $U_{A^{+d}}$ and $U_{B^{+d}}$ act on $A^{+d}$ and $B^{+d}$ respectively. Mutual information is invariant under local unitaries and non-increasing under partial trace, so $I_{U\rho U^\dagger}(A,B)\le I_\rho(A^{+d}:B^{+d})$. 

Applying the same argument to $U\rho U^\dagger$ with the inverse circuit $U^\dagger$ yields the lower bound $I_\rho(A^{-d}:B^{-d})\le I_{U\rho U^\dagger}(A:B)$. This completes the proof.
\end{proof}

The following Lemma converts the triangle inequality for the Bures angle to a bound on fidelity.
\begin{lemma}\label{thm:fidelity_triangle_inequality}
If $F(\rho,\sigma)\le\delta_1$ and $F(\sigma,\tau)\ge 1-\delta_2$, then $F(\rho,\tau)\le\delta_1+\sqrt{2\delta_2}$. 
\end{lemma}
\begin{proof}
Define the Bures angle as $\theta(\rho,\sigma)=\arccos F(\rho,\sigma)$, which satisfies the triangle inequality~\cite{Nielsen_Chuang_2010} $\theta(\rho,\tau)\le\theta(\rho,\sigma)+\theta(\sigma,\tau)$. Taking the cosine gives
\begin{equation}
F(\rho,\tau)\le F(\rho,\sigma)F(\sigma,\tau)+\sqrt{1-F^2(\rho,\sigma)}\sqrt{1-F^2(\sigma,\tau)}.
\end{equation}
\end{proof}

\section{List of notations}
\begin{itemize}
\item $\sigma$: reference state where entanglement bootstrap methods are implemented.
\item $\Sigma(\Om)$: information convex set of a region $\Om$.
\item $\calG(S)$: generators of the stabilizer group $S$.
\item $S_A$: the stabilizer group generated by local stabilizers in $A$.
\item $L_A$: the Pauli operators on $A$ that commute with $S_A$.
\item $S_A(B)$: the stabilizers in $S_A$ that are fully supported on $B$.
\item $L_A(B)$: the operators in $L_A$ that are fully supported on $B$.
\item $S_r(A)$: the stabilizers of the reference state that are fully supported on $A$.
\item $\calS$: the convex hull of (pure) stabilizer states
\item $\SP$: the set of stabilizer projection states.
\item $D$: spatial dimension.
\item $d$: circuit depth.
\item $\mathrm{D}$: Hilbert space dimension of a patch.
\end{itemize}

\bibliography{reference}
\end{document}